\documentclass[11pt]{article}
\usepackage[letterpaper, margin=1.05in]{geometry}
\usepackage[utf8]{inputenc}
\usepackage[T1]{fontenc}

\usepackage{inconsolata}
\usepackage{libertine}

\usepackage{amsthm}
\usepackage{amssymb}
\usepackage{amsmath}
\usepackage{mathtools}
\usepackage{mathrsfs}
\usepackage{thmtools,thm-restate}

\usepackage{tikz}
\usepackage{graphicx}
\graphicspath{ {./images/} }

\usepackage{algorithm}
\usepackage[noend]{algpseudocode}
\usepackage{caption}
\usepackage{subcaption}
\usepackage[noadjust]{cite}
\usepackage{enumitem}
\usepackage{xspace}
\usepackage{xcolor}
\usepackage{xparse}
\usepackage{xfrac}
\usepackage{todonotes}
\usepackage{textpos}

\usepackage[colorlinks = true,linkcolor = blue,urlcolor = blue, citecolor = blue]{hyperref}

\usepackage[capitalize, nameinlink]{cleveref}

\newtheorem{theorem}{Theorem}[section]

\newtheorem{lemma}[theorem]{Lemma}
\newtheorem{observation}[theorem]{Observation}
\newtheorem{corollary}[theorem]{Corollary}

\newtheorem{definition}[theorem]{Definition}

\crefname{claim}{Claim}{Claims}
\newtheorem{claim}[theorem]{Claim}
\newcommand{\cqed}{\ensuremath{\lhd}}
\makeatletter
\newenvironment{claimproof}{\par
	\pushQED{\cqed}%
	\normalfont \topsep6\p@\@plus6\p@\relax
	\trivlist
	\item\relax
	{\itshape
		Proof of the claim\@addpunct{.}}\hspace\labelsep\ignorespaces
}{%
	\hfill\popQED\endtrivlist\@endpefalse
}
\makeatother

\newcommand{\ceil}[1]{\left\lceil #1 \right\rceil}

\newcommand{\bag}{\ensuremath{\mathtt{bag}}}

\newcommand{\Oh}{\mathcal{O}}
\newcommand{\parent}{\mathsf{parent}}
\newcommand{\tw}{\mathrm{tw}}
\newcommand{\td}{\mathrm{td}}

\newcommand{\cc}{\mathrm{cc}}

\newcommand{\D}{\mathbb{D}}

\newcommand{\N}{\mathbb{N}}

\newcommand{\Cc}{\mathscr{C}}

\newcommand{\Nc}{\mathcal{N}}

\newcommand{\R}{\mathbb{R}}

\newcommand{\Cfrak}{\ensuremath{\mathfrak{C}}\xspace}

\renewcommand{\leq}{\leqslant}
\renewcommand{\geq}{\geqslant}

\renewcommand{\le}{\leqslant}
\renewcommand{\ge}{\geqslant}

\newcommand{\eps}{\varepsilon}
\newcommand{\OPT}{\mathsf{OPT}}
\newcommand{\wei}{\mathbf{w}}

\renewcommand{\preceq}{\preccurlyeq}

\newcommand{\anc}{\mathsf{anc}}
\newcommand{\desc}{\mathsf{desc}}
\newcommand{\Reach}{\mathsf{Reach}}
\newcommand{\rev}{\mathsf{rev}}
\newcommand{\cost}{\mathsf{cost}}
\newcommand{\supply}{\mathsf{supply}}
\newcommand{\demand}{\mathsf{demand}}
\newcommand{\Supply}{\mathsf{Supply}}
\newcommand{\Demand}{\mathsf{Demand}}

\newcommand{\interaction}{\mathsf{interaction}}

\newcommand{\problemname}[1]{{\sc #1}\xspace}
\newcommand{\MaxInd}{\problemname{Max Weight Independent Set}}
\newcommand{\CSP}{\problemname{Max Weight Nullary 2CSP}}
\newcommand{\CSPs}{\problemname{2CSP}}

\newcommand{\MinDom}{\problemname{Min Weight Dominating Set}}
\newcommand{\MinGeneralizedDom}{\problemname{Min Weight Generalized Domination}}

\newcommand{\MWIS}{\mathsf{IS}}
\newcommand{\MWDS}{\mathsf{DS}}

\newcommand{\commandname}[1]{{\sf #1}}
\newcommand{\AddVertex}{\commandname{AddVertex}}
\newcommand{\AddEdge}{\commandname{AddEdge}}
\newcommand{\RemoveEdge}{\commandname{RemoveEdge}}
\newcommand{\UpdateRevenue}{\commandname{UpdateRevenue}}
\newcommand{\UpdateWeight}{\commandname{UpdateWeight}}
\newcommand{\QueryMWIS}{\commandname{QueryMWIS}}
\newcommand{\QueryMWDS}{\commandname{QueryMWDS}}
\newcommand{\UpdateCost}{\commandname{UpdateCost}}
\newcommand{\Clear}{\commandname{Clear}}

\newcommand{\Gold}{G^{\mathrm{old}}}
\newcommand{\Ginit}{G^{\mathrm{init}}}
\newcommand{\Gcur}{G^{\mathrm{cur}}}
\newcommand{\Iold}{I^{\mathrm{old}}}
\newcommand{\Iinit}{I^{\mathrm{init}}}
\newcommand{\Icur}{I^{\mathrm{cur}}}
\newcommand{\Vcur}{V^{\mathrm{cur}}}

\newcommand{\shat}{\widehat{s}}
\newcommand{\dhat}{\widehat{d}}

\renewcommand{\setminus}{-}

\begin{document}
\title{Fully dynamic approximation schemes on planar and apex-minor-free graphs%
\thanks{This work is a~part of the project BOBR (WN, MP, MS) that has received funding from the European Research Council (ERC) under the European Union's Horizon 2020 research and innovation programme (grant agreement No.\ 948057). Tuukka Korhonen was supported by the Research Council of Norway via the project
BWCA (grant no. 314528).}}
\author{Tuukka Korhonen\thanks{Department of Informatics, University of Bergen, Norway (\texttt{tuukka.korhonen@uib.no})}  \and
Wojciech Nadara\thanks{Institute of Informatics, University of Warsaw, Poland (\texttt{w.nadara@mimuw.edu.pl})} \and
Michał Pilipczuk\thanks{Institute of Informatics, University of Warsaw, Poland (\texttt{michal.pilipczuk@mimuw.edu.pl})} \and
Marek Sokołowski\thanks{Institute of Informatics, University of Warsaw, Poland (\texttt{marek.sokolowski@mimuw.edu.pl})}}
\date{}
\maketitle
\thispagestyle{empty}

 \begin{textblock}{20}(-1.9, 8.2)
  \includegraphics[width=40px]{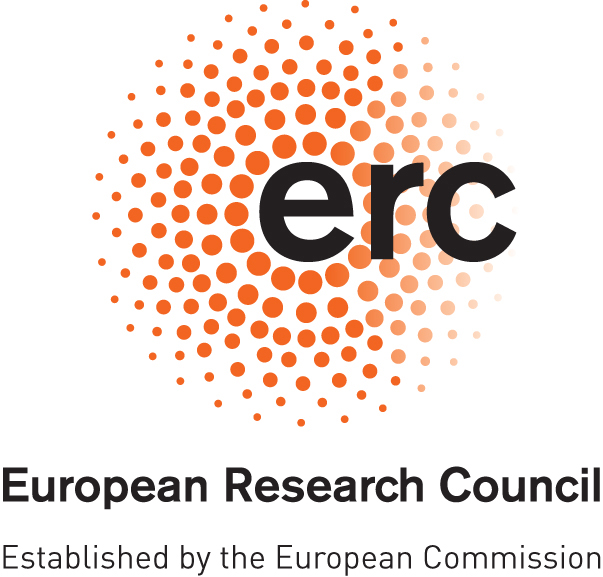}%
 \end{textblock}
 \begin{textblock}{20}(-2.15, 8.5)
  \includegraphics[width=60px]{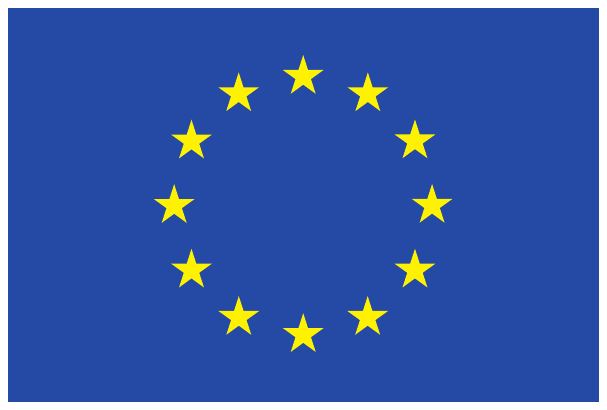}%
 \end{textblock}

\begin{abstract}
The classic technique of Baker~[J.~ACM~'94] is the most fundamental approach for designing approximation schemes on planar, or more generally topologically-constrained graphs, and it has been applied in a myriad of different variants and settings throughout the last 30 years. In this work we propose a dynamic variant of Baker's technique, where instead of finding an approximate solution in a given static graph, the task is to design a data structure for maintaining an approximate solution in a fully dynamic graph, that is, a graph that is changing over time by edge deletions and edge insertions. Specifically, we address the two most basic problems --- {\sc{Maximum Weight Independent Set}} and {\sc{Minimum Weight Dominating Set}} --- and we prove the following: for a fully dynamic $n$-vertex planar graph $G$, one can
\begin{itemize}[nosep]
 \item maintain a $(1-\eps)$-approximation of the maximum weight of an independent set in $G$ with amortized update time $f(\eps)\cdot n^{o(1)}$; and,
 \item under the additional assumption that the maximum degree of the graph is bounded at all times by a constant, also maintain a $(1+\eps)$-approximation of the minimum weight of a dominating set in $G$ with amortized update time $f(\eps)\cdot n^{o(1)}$.
\end{itemize}
In both cases, $f(\eps)$ is doubly-exponential in $\mathrm{poly}(1/\eps)$ and the data structure can be initialized in time $f(\eps)\cdot n^{1+o(1)}$. All our results in fact hold in the larger generality of any graph class that excludes a fixed apex-graph as a minor.

\end{abstract}

\newpage
\pagenumbering{arabic}

\section{Introduction}\label{sec:intro}


{\em{Baker's technique}}, also known as {\em{shifting}} or {\em{layering}} and proposed by Baker in 1994~\cite{Baker94}, is the most fundamental technique for designing  {\em{efficient polynomial-time approximation schemes}} (EPTASes) --- approximation schemes with running time of the form $f(\eps)\cdot n^{\Oh(1)}$ --- for problems on planar graphs, or more generally in topologically restricted graph classes. In terms of restrictions on the input graph, the basic approach works as long as the considered class of graphs $\Cc$ has {\em{bounded local treewidth}}: the treewidth of any connected $G\in\Cc$ is bounded by a function of the radius of $G$. This includes planar graphs (with the function being $3$ times the radius) and, as proved by Eppstein~\cite{Eppstein00}, also all {\em{apex-minor-free}} graph classes: classes that exclude a fixed {\em{apex graph}} --- a graph that can be made planar by removing one vertex --- as a minor. However, further generalizations of the technique apply also in the setting of $H$-minor-free graphs for any fixed $H$~\cite{Grohe03}, and even beyond~\cite{Dvorak18,Dvorak20,Dvorak22}. In terms of versatility, the range of applicability of Baker's technique is surprisingly wide, and can be even captured by meta-theorems concerning optimization problems expressible in first-order logic~\cite{DawarGKS06,Dvorak22}. Importantly, the spectrum of applicability includes maximization problems of packing nature such as {\sc{Maximum Weight Independent Set}} --- in a vertex-weighted graph $G$, find a set $I$ of maximum possible weight consisting of pairwise non-adjacent vertices --- and minimization problems of covering nature such as {\sc{Minimum Weight Dominating Set}} --- find a~set of vertices $D$ of minimum possible weight such that every vertex outside of $D$ has~a~neighbor~in~$D$.

During the last $30$ years, the basic principle behind Baker's technique has inspired countless important developments in the area of approximation schemes, see e.g.~\cite{Cohen-AddadPP19,EisenstatKM14,Fox-EpsteinKS19,KleinPR93} besides the works mentioned above. It has also found multiple important applications in parameterized algorithms, see e.g. the discussion in~\cite[Section~7.7.3]{platypus}. 

\paragraph*{Our contribution.} In this work we investigate the following question asked recently by Korhonen, Majewski, Nadara, Pilipczuk, and Sokołowski~\cite{DBLP:journals/corr/abs-2304-01744}: to what extent Baker's technique can be applied in the context of dynamic algorithms. That is, instead of just looking for an approximate solution in a static input graph, we assume that the graph is fully dynamic --- changes over time by edge insertions and deletions --- and we would like to maintain an approximate value of the optimum under such updates. For the sake of focus, we restrict attention to the aforementioned two basic problems: {\sc{Maximum Weight Independent Set}} and {\sc{Minimum Weight Dominating Set}}.

Consider the following notion. A {\em{fully dynamic graph data structure}} is a data structure that maintains a graph $G$ with nonnegative weights on vertices under the following~updates:
\begin{itemize}
 \item $\AddEdge(u,v)$, $\RemoveEdge(u,v)$: add or remove an edge between vertices $u,v$; and
 \item $\UpdateWeight(u,\alpha)$: change the weight of a vertex $u$ to $\alpha\in \mathbb{R}_{\geq 0}$.
\end{itemize}
We shall say that such a data structure is {\em{$\Cc$-restricted}}, for a graph class $\Cc$, if it works under the promise that at all times, the graph stored in the data structure belongs to $\Cc$.

With this definition in place, we can state our main result.

\begin{theorem}\label{thm:main}
 Let $\Cc$ be a~fixed apex-minor-free class of graphs and let $\eps>0$. Then there exists a $\Cc$-restricted fully dynamic graph data structure that in addition to maintaining a graph $G\in \Cc$, supports the following queries:
 \begin{itemize}
  \item $\QueryMWIS()$: output a nonnegative real $p$ satisfying $(1-\eps)\OPT_{\MWIS}\leq p\leq \OPT_{\MWIS}$, where $\OPT_{\MWIS}$ is the maximum weight of an independent set in $G$; and
  \item $\QueryMWDS()$: output a nonnegative real $p$ satisfying $\OPT_{\MWDS}\leq p\leq (1+\eps)\OPT_{\MWDS}$, where $\OPT_{\MWDS}$ is the minimum weight of a dominating set in $G$. This query is supported only under the additional assumption that at all times, the maximum degree of $G$ is bounded by a constant $\Delta$.
 \end{itemize}
 The initialization time on a given $n$-vertex graph $G\in \Cc$ is $f(\eps)\cdot n^{1+o(1)}$, and each update takes amortized time $f(\eps)\cdot n^{o(1)}$, where $f(\eps)$ is doubly-exponential in $\Oh(1/\eps^2)$. Each query takes $\Oh(1)$ time.
\end{theorem}

Note that in general graphs, {\sc{Maximum Weight Independent Set}} and {\sc{Minimum Weight Dominating Set}} do not admit EPTASes in the static setting if $\mathsf{P} \neq \mathsf{NP}$, even with a~restriction on the maximum degree of a~vertex of a~graph~\cite{DBLP:journals/iandc/ChlebikC08,DBLP:conf/stoc/Trevisan01}.
Hence a~structural restriction of the class $\Cc$ in \cref{thm:main} is necessary.

To the best of our knowledge, \Cref{thm:main} presents the first non-trivial data structure for approximation schemes on planar graphs in a fully dynamic setting\footnote{We remark that for the {\sc{Maximum Independent Set}} problem, where each vertex carries a unit weight, there is a simple data structure with amortized update time $2^{\Oh(1/\eps)}$ working as follows. Once every $\eps n/8$ updates recompute a $(1-\eps/2)$-approximate solution from scratch using Baker's technique. Between the recomputations, whenever an edge is added to the graph, remove any of its endpoints from the solution provided both were included. This works because by the $4$-Color Theorem, the optimum solution has always size at least $n/4$, and within $\eps n/8$ updates the value of the optimum may change by at most $\eps n/8$. This is why we focus on the weighted variant of the problem, to make it non-trivial.}. Hence, we hope that it may open multiple new avenues for future research. Here are three questions that immediately come to mind:
\begin{itemize}
 \item \Cref{thm:main} applies only to apex-minor-free classes; can it be extended to $H$-minor-free classes for any fixed $H$? The main obstacle here is that we are not aware of any approach that yields a (static) EPTAS for the considered problems in $H$-minor-free graphs with a near-linear running time dependency on graph size; and this would be implied by such an extension.
 \item \Cref{thm:main} tackles only the {\sc{Maximum Weight Independent Set}} and {\sc{Minimum Weight Dominating Set}} problems; can it be extended to other problems amenable to Baker's technique, for instance the first-order expressible optimization problems considered in~\cite{DawarGKS06,Dvorak22}? As the reader will see, the proof of \Cref{thm:main} is considerably more difficult and delicate than that of the basic Baker's technique. At this point, even extending the result for {\sc{Minimum Weight Dominating Set}} beyond the regime of bounded-degree graphs seems unclear.
 \item The $n^{o(1)}$ factor in the bound on the update time of our data structure is actually $n^{\Oh\left(\frac{\log \log \log n}{\sqrt{\log \log n}}\right)}$. This means that the amortized update time is indeed subpolynomial, but barely. It would be interesting to obtain better bounds, ideally polylogarithmic in $n$. Designing data structures with worst-case guarantees is another interesting question.
\end{itemize}  
We remark that while Korhonen et al.~\cite{DBLP:journals/corr/abs-2304-01744} asked their question in the context of their work on a data structure for maintaining dynamic tree decompositions of bounded width in amortized subpolynomial time, in the proof of \Cref{thm:main} we do not use their data structure at all. It would be also interesting to see if the approach presented here can be combined with the advances of~\cite{DBLP:journals/corr/abs-2304-01744}.

\newcommand{\Tt}{\mathcal{T}}

\paragraph*{Our techniques.} The main idea behind the proof of \cref{thm:main} is to maintain a tree of data structures, where each data structure is responsible for some minor $H$ of $G$ and has $\Oh(1/\eps)$ child data structures, responsible for the $\Oh(1/\eps)$ choices of offsets in the Baker scheme applied to $H$. While we eventually create a tree of data structures with $L=\Theta\left(\eps\cdot \frac{\log \log n}{\log \log \log n}\right)$ levels, let us explain the construction for $L=2$ levels; this corresponds to achieving amortized update time $f(\eps)\cdot \widetilde{\Oh}(\sqrt{n})$. Also, let us focus on the {\sc{Maximum Weight Independent Set}} problem.

Let $k\coloneqq \lceil 1/\eps\rceil$.
There will be a parent data structure $\D^{\textrm{main}}$ responsible for the whole graph~$G$. Upon the initialization of $\D^{\textrm{main}}$, we apply Baker's scheme to $G$: we compute a partition $V_0,\ldots,V_{k-1}$ of the vertex set of $G$ so that $V_i$ comprises vertices whose distance from some fixed vertex $s$ is congruent to $i$ modulo $k$. (If $G$ is disconnected, every connected component has its own source vertex $s$.) Standard analysis of Baker's technique yields that the treewidth of the graph $G_i\coloneqq G-V_i$ is at most $3k$, so we compute a tree decomposition $\Tt_i$ of $G_i$ of width $\Oh(k)$ and depth $\Oh(\log n)$ using classic results~\cite{bodlaender-hagerup}. By dynamic programming on $\Tt_i$, we can understand $G-V_i$ completely. In particular, we may compute the optimum weight of an independent set in $G_i$, so that the maximum among the computed values is a $(1-\eps)$-approximation of the maximum weight of an independent set in $G$.

At this point every graph $G_i$ has treewidth at most $3k$, but this may quickly cease to be the case once some updates arrive. Therefore, for every $i\in \{0,1,\ldots,k-1\}$ we create a child data structure $\D_i$ that is responsible for handling the graph $G_i$. In the data structure $\D_i$, vertices of $G_i$ are partitioned into an initially empty {\em{stash}} $Z$, which contains all vertices involved in any updates made so far, and the remaining vertices. We maintain the following invariant $(\bigstar)$: every connected component of $G_i-Z$ has neighborhood of size $\Oh(k)$. Data structure $\D_i$ maintains the {\em{compressed graph}} $H_i$, which is a minor of $G_i$ obtained as follows: (i) contract every connected component of $G_i-Z$ to a single vertex, and (ii) identify all vertices corresponding to contracted components with the same neighborhood into single vertices. Since $Z$ is initially empty, $H_i$ is initially a single-vertex graph. When an update affects $G_i$, say an insertion of an edge $uv$, we add to $Z$ both $u$ and $v$ together with the $\Oh(k\log n)$ vertices contained in all the ancestor bags\footnote{The actual presentation differs here in that we  work with the notion of an {\em{elimination forest}} instead of a tree decomposition. However, this would be the tree decomposition equivalent of this step.} of the top-most bags of $\Tt_i$ that contain $u$ and $v$; this way we preserve invariant~$(\bigstar)$. Every addition of a~vertex to $Z$ results in ``uncompressing'' a part of $H_i$, which can be done efficiently and adds only $\Oh(k)$ new~vertices~to~$H_i$.

The data structure $\D_i$ maintains also an approximation of the maximum weight of an independent set in $G_i$. This is done by maintaining an approximate optimum for an auxiliary maximization problem on the compressed graph $H_i$, which we find convenient to phrase in the language of valued {\sc{2CSP}}s, and which can be considered a ``contracted'' variant of {\sc{Maximum Weight Independent Set}}. More precisely, every vertex of $H_i$ that is also a vertex of $G_i$ has a domain consisting of two values --- taken or not taken --- while every vertex that resulted from a contraction of some connected components of $G_i-Z$ has a larger domain, corresponding to possible interactions between those components and the rest of the graph. The revenues provided by the contracted vertices for different elements of their domains reflect the maximum weights of independent sets that can be achieved within contracted connected components for different interactions with the rest of the graph. And these maximum weights can be read from the dynamic programming tables in $\Tt_i$ computed upon the initialization of $\D^{\textrm{main}}$, because vertices of $G_i-Z$ have not yet been touched by updates. The approximate optimum to the constructed auxiliary instance of {\sc{2CSP}} is computed {\em{from scratch}} upon every update in $H_i$, in time $f(\eps)\cdot |V(H_i)|$ using Baker's~technique~in~$H_i$. 

Thus, every update to $G$ is relayed by $\D^{\textrm{main}}$ to $k$ data structures $\D_i$, and within each $\D_i$ we apply $k^{\Oh(1)} \log n$ updates to $H_i$, each consisting of running Baker's scheme in time $f(\eps)\cdot |V(H_i)|$. The crucial idea is to {\em{reset}} the data structure $\D^{\textrm{main}}$ every $\sqrt{n}$ updates, by reconstructing it from scratch in time $f(\eps)\cdot n$, so that graphs $H_i$ never grow too large. After every reconstruction of $\D^{\textrm{main}}$, all graphs $H_i$ consist of single vertices, so throughout the next $\sqrt{n}$ updates they can grow to size at most $k^{\Oh(1)}\sqrt{n} \log n$, because every $H_i$ grows by at most $k^{\Oh(1)}\log n$ vertices at each update. This means that running Baker's scheme upon every update to any $H_i$ will take worst-case time $f(\eps)\cdot \widetilde{\Oh}(\sqrt{n})$, while the amortized time complexity of resetting the data structure $\D^{\textrm{main}}$ is $f(\eps)\cdot \widetilde{\Oh}(\sqrt{n})$ as well.

Our data structure for the {\sc{Maximum Weight Independent Set}} problem implements the natural generalization of the idea presented above to more than two levels. Specifically, we use $L=\Theta\left(\eps\cdot \frac{\log \log n}{\log \log \log n}\right)$ levels, where every data structure $\D$ at level $i$ gets reset every $n^{1-\frac{i}{L}}$ updates relayed to $\D$. One technical detail that requires attention are the domains in the auxiliary {\sc{2CSP}} instances: every compression step increases the maximum domain size from, say, $\Delta$ to $\Delta^{\Oh(k)}$, so one needs to carefully choose the number of layers $L$ so that it is super-constant in $n$, but small enough so that the domain sizes are kept subpolynomial in $n$.

The data structure for the {\sc{Minimum Weight Dominating Set}} problem is similar in spirit. As a~matter of fact, the standard way of applying Baker's technique to this problem, by allowing double-paying in an $\eps$-fraction of the layers, seems difficult to translate to the dynamic setting. Instead, we design a~different way of applying Baker's technique to {\sc{Minimum Weight Dominating Set}}, which interestingly relies on {\em{under-approximating}} the minimum weight of a dominating set, rather than over-approximating. There are several technical issues that arise when applying the approach presented above in the setting of dominating sets: in particular when a vertex $u$ is touched by an update within any constituent data structure $\D$, we need to update the suitable information about both $u$ and the neighbors of $u$ in the graph. This is why in the result for {\sc{Minimum Weight Dominating Set}} we resort to the regime of graphs with~bounded maximum degree.

Finally, let us remark that in the approach sketched above, the key properties kept by each of the constituent data structures are the invariant $(\bigstar)$ and the invariant that the removal of the stash breaks the graph into connected components of treewidth $\Oh(k)$. This exactly means that each of those components will be an {\em{$\Oh(k)$-protrusion}}. Protrusion-based arguments are by now a standard methodology in the design of parameterized and kernelization algorithms on planar and topologically-restricted graphs, see e.g.~\cite{BodlaenderFLPST16,FominLST20} or~\cite[Chapter~15]{squirrel}. In this work we show how these ideas can be helpful also in the setting of dynamic approximation algorithms.

\section{Preliminaries}
\label{sec:prelims}

Given a~set $X$ and a~family of sets $\{A_x \,\colon\, x \in X\}$, we define $\prod_{x \in X} A_x$ as the Cartesian product of the sets $A_x$.
Then $f \in \prod_{x \in X} A_x$ is a~function mapping each $x \in X$ to an~element $f(x) \in A_x$.

\paragraph*{Graphs.}
All graphs in this work are undirected; moreover, they are simple unless explicitly mentioned.
For a~graph $G$ and a~vertex $u$ of $G$, $N_G(u)$ is the set of neighbors of $u$ in $G$, and $N_G[u] = N_G(u) \cup \{u\}$; if $G$ is known from the context, we write $N(u)$ and $N[u]$.
If $S \subseteq V(G)$, then let $N_G[S] \coloneqq \bigcup_{u \in S} N_G[u]$ and $N_G(S) \coloneqq N_G[S] \setminus S$. For a~graph $G$, we denote by $\cc(G)$ the partitioning of $V(G)$ into the set of connected components.
If $S \subseteq V(G)$, then we denote by $G[S]$ the subgraph of $G$ induced by $S$ and by $G \setminus S$ the subgraph of $G$ induced by $V(G) \setminus S$.
If $G$ is connected and $u,v \in V(G)$, we denote by $\mathsf{dist}_G(u,v)$ the number of edges in a shortest path between $u$ and $v$.
The radius of $G$ is defined as $\min_{s \in V(G)} \max_{v \in V(G)} \mathsf{dist}_G(s, v)$.

A {\em{tree decomposition}} of a graph $G$ consists of a forest $F$ and a function $\bag$ that with each node $x$ of $F$ associates its bag $\bag(x)\subseteq V(G)$ so that: (i) for each edge $uv$ of $G$ there exists $x\in V(F)$ with $\{u,v\}\subseteq \bag(x)$, and (ii) for each vertex $u$ of $G$, the set $\{x\in V(F)~|~u\in \bag(x)\}$ induces in~$F$ a~non-empty tree. The {\em{width}} of $(F,\bag)$ is the maximum bag size minus $1$, and the {\em{treewidth}} of~$G$, denoted $\tw(G)$, is the minimum possible width of a tree decomposition of~$G$.

In this work we will mostly not use tree decompositions as defined above, but we will instead rely on the related notion of {\em{elimination forests}}, which underlie the parameter {\em{treedepth}}. First, we need some definitions on rooted forests.
Whenever $F$ is a rooted forest, by $\preceq_F$ we denote the ancestor relation in $F$: for $u,v\in V(F)$, $u\preceq_F v$ if $u$ is on the unique path in $F$ from $v$ to a root of $F$. For a vertex $u$ of $F$, 
$\parent_F(u)$ is the parent of $u$ in $F$ (or $\bot$ if $u$ is a root), and we denote the ancestors and the descendants of $u$ as $\anc_F[u]\coloneqq \{v\colon v\preceq_F u\}$ and $\desc_F[u]\coloneqq \{v\colon u\preceq_F v\}$, respectively. Note that thus, every vertex is both an ancestor and a descendant of itself, i.e. $\anc_F[u]\cap \desc_F[u]=\{u\}$. We can generalize $\anc_F[u]$ to $\anc_F[A]$ for $A \subseteq V(F)$ by defining $\anc_F[A] = \bigcup_{u \in A} \anc_F[u]$; this set will be called the \emph{ancestor closure} of $A$ in $F$. Sets $A$ closed under the taking ancestors, i.e. with $A=\anc_F[A]$, will be called \emph{prefixes} of $F$.
The {\em{height}} of a forest is the maximum number of vertices on a root-to-leaf path in $F$.
An {\em{elimination forest}} of a graph $G$ is a rooted forest $F$ with $V(F)=V(G)$ such that for every edge $uv$ of $G$, we have $u\preceq_F v$ or $v\preceq_F u$. The {\em{treedepth}} of a graph $G$, denoted $\td(G)$, is the minimum height of an elimination~forest~of~$G$.

If $F$ is an elimination forest of $G$, then with every vertex $u$ we can associate its {\em{reachability set}} defined as
$\Reach_F(u)\coloneqq N(\desc_F[u])$.
Note that as $F$ is an elimination forest, $\Reach_F(u)$ consists of strict ancestors of $u$, that is, $\Reach_F(u)\subseteq \anc_F[u]\setminus \{u\}$. So in other words, $\Reach_F(u)$ comprises all strict ancestors of $u$ that have a neighbor among the descendants of $u$.
It can be easily seen that if $F$ is an elimination forest of $G$, then endowing $F$ with a bag function $u\mapsto \{u\}\cup \Reach_F(u)$ yields a tree decomposition of $G$. As proved in~\cite{BojanczykP22}, for every graph $G$ there is an elimination forest $F$ of $G$ such that a tree decomposition obtained from $F$ in this way has optimum width, that is, width equal to the treewidth~of~$G$.

\paragraph*{Multigraphs.}
In the case of Dominating Set and its generalizations, it will be convenient to work with multigraphs $G = (V, E)$.
We assume that in a~multigraph, there may be multiple edges connecting the same pair of vertices (also called \emph{parallel edges}), but there are no self-loops.
We distinguish different edges connecting the same pair of vertices -- for example, we can assume that the edges of the graph have pairwise different labels.

For a~vertex $v \in V(G)$, let $\delta(v)$ denote the set of edges with one endpoint in $v$.
We define the degree of $v$ as $\deg(v) = |\delta(v)|$.
For two disjoint sets of vertices $A, B \subseteq V(G)$, we use the notation $E(A, B)$ to denote the set of edges with one endpoint in $A$ and the other in $B$.
Given a~set of vertices $S$, to \emph{collapse} $S$ in $V(G)$ is to identify all vertices of $S$ to a~single vertex $v_S$, remove all edges with both endpoints in $S$, and for all other edges with an~endpoint in $S$, reroute the endpoint to $v_S$. Note that this process might produce a~multigraph from a~simple graph.


\paragraph*{Apex graphs.}
A~simple graph $H$ is called an~\emph{apex graph} if removing a~vertex from $H$ causes it to become planar.
We say that a~graph $G$ contains a~graph $H$ as a~minor if $H$ can be constructed from $G$ by a~sequence of vertex removals, edge removals and edge contractions.
We then say that a~class of graphs $\Cc$ is \emph{apex-minor-free} if $\Cc$ is \emph{minor-closed} (i.e., for every graph $G \in \Cc$, all minors of $G$ also belong to $\Cc$), and $\Cc$ excludes a~fixed apex graph $H$.
For instance, the class of planar graphs is apex-minor-free.

We will use the~result of Demaine and Hajiaghayi~\cite{ApexFreeLocBdTw} that apex-minor-free classes of graphs have \emph{linearly locally bounded treewidth}, which improved upon the earlier work of Eppstein~\cite{Eppstein00}.

%

\begin{theorem}[\cite{ApexFreeLocBdTw}]
  \label{thm:linearly-locally-bounded-tw}
  For every apex-minor-free class $\Cc$ of graphs there exists a~constant $\kappa_\Cc > 0$ such that for every integer $r \geq 1$ and graph $G \in \Cc$ of radius at most $r$, we have $\tw(G) \leq \kappa_\Cc \cdot r$.
\end{theorem}

Apex-minor-free classes of graphs (and more generally, minor-free classes) have a~bounded ratio of the number of edges to the number of vertices \cite{DBLP:journals/combinatorica/Kostochka84}; that is, for any apex-minor-free class of graphs $\Cc$, there exists a~constant $\rho_\Cc > 0$ so that for every graph $G \in \Cc$, it holds that $|E(G)| \leq \rho_\Cc \cdot |V(G)|$.

Additionally, the following facts from the theory of sparse graphs will prove useful for us.
Since every apex-minor-free class of graphs (and more generally, every minor-closed class excluding at least one graph) has bounded expansion \cite{DBLP:journals/ejc/NesetrilM08}, it has bounded \emph{neighborhood complexity}:

\begin{theorem}[\cite{NeiComp1}]
  \label{thm:neighborhood-complexity-orig}
  There exists a~constant $\nu_\Cc > 0$ such that for every $G \in \Cc$ and nonempty $X \subseteq V(G)$, the number of different neighborhoods of vertices of $G$ on $X$ is bounded by $\nu_\Cc \cdot |X|$.
  That is,
  \[ \left| \{ N(v) \cap X \,\colon\, v \in V(G) \} \right| \leq \nu_\Cc \cdot |X|. \]
\end{theorem}

Therefore, we have that:

\begin{corollary}
  \label{cor:neighborhood-complexity-components}
  If $\Cc$ is apex-minor-free, $G \in \Cc$ and $X \subseteq V(G)$ is nonempty, then
  \[ \left| \{ N(C) \,\colon\, C \in \cc(G \setminus X) \} \right| \leq \nu_\Cc \cdot |X|. \]
\end{corollary}
\begin{proof}
  Produce a~graph $H \in \Cc$ by contracting each connected component of $G \setminus X$ to a~single vertex and apply \cref{thm:neighborhood-complexity-orig} to $H$ and $X$.
\end{proof}

\paragraph*{Generalizing independent sets.}
In this paper, we shall work with a generalization of the \MaxInd problem that is most conveniently phrased in the language of weighted CSPs. An instance $I$ of \CSP consists of:
\begin{itemize}
 \item a graph $G$, called the {\em{Gaifman graph}};
 \item for every vertex $u\in V(G)$, a finite {\em{domain}} $D_u$ and a {\em{revenue function}} $\rev_u\colon D_u\to \R_{\geq 0}$; and
 \item for every edge $uv\in E(G)$, a {\em{constraint}} $C_{uv}\subseteq D_u\times D_v$. 
\end{itemize}
We require that each domain $D_u$ contains a special value $0$ whose revenue is $0$, that is, $\rev_u(0)=0$. Moreover, the value $0$ is always allowed in constraints: for every edge $uv$, we have $\{0\}\times D_v\subseteq C_{uv}$ and $D_u\times \{0\}\subseteq C_{uv}$. We can assume that at least one domain is of size at least two, as otherwise the problem is trivial. We will concisely denote that $I = (G, D, \rev, C)$, where $D$ is the set of all domains, $\rev$ is the set of revenue functions and $C$ is the set of all constraints.

A {\em{solution}} to a \CSP instance as above is a mapping $\phi$ that with each vertex $u$ of $G$ associates a value $\phi(u)\in D_u$ so that $(\phi(u),\phi(v))\in C_{uv}$ for every edge $uv$ of $G$. Note that such a solution always exists, as one can map every vertex $u$ to $0\in D_u$. The {\em{revenue}} of a solution $\phi$ is defined~as
\[ \rev(\phi)\coloneqq \sum_{u\in V(G)} \rev_u(\phi(u)). \]
In \CSP one is asked to find a solution with the maximum possible revenue. 

For our purposes, the domains $D_u$ can be assumed to be small (i.e., bounded in size by $f(\eps) \cdot n^{o(1)}$ for some function $f$).
With this assumption in mind, in our time complexity analysis we will omit the time required to store and manipulate the domains and the constraints of the CSP.

We may model the \MaxInd problem in the language of \CSP. Let $G$ be a graph and $\wei\colon V(G)\to \R_{\geq 0}$ be a weight function on the vertices of $G$. Then construct a \CSP instance with Gaifman graph $G$ as follows:
\begin{itemize}
 \item For every vertex $u$ of $G$, we set $D_u=\{0,1\}$, $\rev_u(0)=0$, and $\rev_u(1)=\wei(u)$.
 \item For every edge $uv$ of $G$, we set $C_{uv}=\{(0,0),(0,1),(1,0)\}$.
\end{itemize}
It is straightforward to see that the maximum revenue of a solution to the instance described above is equal to the maximum weight of an independent set in $G$.

For convenience, given an~instance $I = (G, D, \rev, C)$ of \CSP and a vertex subset $Y \subseteq V(G)$, we denote by $I[Y]$ the instance \emph{induced} by $Y$. This is the instance obtained from $I$ by removing from $G$ all vertices not in $Y$, and all constraints incident to a~vertex outside of~$Y$. Further, denote $I \setminus Y = I[V(G) \setminus Y]$. Note that any solution to an induced instance can be lifted to a solution to the original instance of the same revenue by mapping all the removed vertices to $0$.
Finally, we denote $V(I) \coloneqq V(G)$ and $E(I) \coloneqq E(G)$.

\paragraph*{Generalizing dominating sets.} Similarly to the case of the {\sc Max Weight Independent Set}, we will work with a generalization of the \MinDom problem defined as follows. An instance of \MinGeneralizedDom consists of:
\begin{itemize}
 \item a multigraph $G$, called the {\em{Gaifman graph}}; and
 \item for every vertex $u$, a finite {\em{domain}} $D_u$, a {\em{cost function}} $\cost_u\colon D_u\to \R_{\geq 0} \cup \{+\infty\}$, a {\em{supply function}} $\supply_u\colon D_u\to 2^{\delta(u)}$ from the domain $D_u$ to the set of all subsets of edges with one endpoint in $u$, and a {\em{demand function}} $\demand_u\colon D_u\to 2^{\delta(u)}$.
\end{itemize}
We additionally require that every domain $D_u$ contains a~state $s_u$ with $\supply_u(s_u) = \delta(u)$ and finite cost.

A {\em{solution}} to an instance as above is a mapping $\phi$ that with each vertex $u$ of $G$ associates a value $\phi(u)\in D_u$, so that the following property holds.
For every edge $e \in E(G)$ with endpoints $u$ and $v$, if $e \in \demand_u(\phi(u))$, then $e \in \supply_v(\phi(v))$; and conversely, if $e \in \demand_v(\phi(v))$, then $e \in \supply_u(\phi(u))$.\footnote{We remark that it might be the case that in a~valid solution, $e$ belongs to both $\demand_u(\phi(u))$ and $\demand_v(\phi(v))$; in this case, we require both $\supply_u(\phi(u))$ and $\supply_v(\phi(v))$ to contain $e$.}

The cost of such a solution is defined as
$$\cost(\phi)=\sum_{u\in V(G)} \cost_u(\phi(u)).$$

The \MinGeneralizedDom problem is to find a solution to the given instance with the minimum possible cost.
Observe that every instance of \MinGeneralizedDom has a~finite-cost solution as a~solution $\phi$ mapping every vertex $u$ to the state $s_u$ is valid.

The \MinDom problem can be modelled using \MinGeneralizedDom as follows. Given a~simple graph $G$ and a weight function $\wei\colon V(G)\to \R_{\geq 0}$, we create an instance $I = (G, D, \cost, \supply, \demand)$, where for every vertex $u$ of $G$ we define the domain $D_u\colon N[u]$ as well as the following cost, supply, and demand functions:
\begin{align*}
\cost_u(v) & = \begin{cases} \wei(u) & \textrm{if }u=v,\\ 0 & \textrm{otherwise;}\end{cases}\\
\supply_u(v) & = \begin{cases} \delta(u)\ & \textrm{if }u=v,\\ \emptyset & \textrm{otherwise;}\end{cases}\\ 
\demand_u(v) & = \begin{cases} \emptyset & \textrm{if }u=v,\\ \{uv\} & \textrm{otherwise.}\end{cases} 
\end{align*}
It is easy to see that the minimum weight of a solution in the instance of \MinGeneralizedDom defined above coincides with the minimum weight of a dominating set in $G$.

Abusing the notation, we will say that $I \in \Cc$ if $G \in \Cc$. Also, we use $V(I)$ and $E(I)$ to mean the set of vertices and the set of edges of $G$.

Next, we will say that a~vertex $u \in V(I)$ is:
\begin{itemize}
  \item $(s, d)$-\emph{meager} for $s, d \geq 1$ if $|\deg(u)| \leq s$ and $|D_u| \leq d$;
  \item \emph{state-monotonous} if for every ordered pair of states $x_1, x_2 \in D_u$, there exists a~state $x$, called the \emph{combination} of $x_1$ with $x_2$, such that
    \begin{align*}
      \cost_u(x) &\leq \cost_u(x_1) + \cost_u(x_2), \\
      \supply_u(x) &= \supply_u(x_1) \cup \supply_u(x_2), \\
      \demand_u(x) &\subseteq \demand_u(x_1).
    \end{align*}
\end{itemize}

Note that the definition of the combination of states above is not commutative due to the fact that the constraint on $\demand_u(x)$ only depends on $x_1$. Thus, the combination of $x_1$ with $x_2$ might be different from the combination of $x_2$ with $x_1$.

We say that an~instance of \MinGeneralizedDom is \emph{$(s,d)$-decent} if every vertex of the instance is $(s,d)$-meager and state-monotonous.
An~easy verification of the definition shows that every instance of \MinGeneralizedDom created from \MinDom on a~graph of maximum degree $\Delta$ is $(\Delta,\Delta+1)$-decent as all vertices of the resulting instance are state-monotonous and $(\Delta, \Delta+1)$-meager.
Our data structure will only operate on $(s, d)$-decent instances, for some small, bounded values of $s$ and $d$.
%


Next, for $A \subseteq V(I)$ and a~valuation $\phi \in \prod_{u \in A} D_u$ of vertices in $A$, we say that the valuation is \emph{locally correct on $A$ to $I$} if, for every edge $e$ with both endpoints $u$, $v$ in $A$, if $e \in \demand_{u}(\phi(u))$, then $e \in \supply_{v}(\phi(v))$.
Note that if $\phi$ is a~correct solution to $I$, then $\phi|_A$ is a~locally correct solution on~$A$~to~$I$.

Given an~instance $I$ of \MinGeneralizedDom and a~vertex $u \in V(I)$, to \emph{relieve} $u$ in $I$ is to produce an~instance $I'$ by setting $\demand_u(x) \gets \emptyset$ for all $x \in D_u$.
In other words, relieving marks $u$ as a~vertex that does not need to be dominated by a~solution to $I'$.
Note that if there exists a~solution to $I$, then a~solution to $I'$ also exists and the minimum cost of such a~solution is less than or equal to the minimum cost of a~solution to $I$.

Finally, given a~set $A \subseteq V(G)$, we define the \emph{$A$-cleared subinstance of $I$}, denoted $\mathsf{Clear}(I; A)$, as the instance produced from $I$ by:
\begin{itemize}
  \item removing all the edges whose both endpoints are in $A$, erasing them also from the corresponding demand and supply sets;
  \item relieving each vertex of $A$ in the resulting instance.
\end{itemize}
The following fact is straightforward.

\begin{observation}
  If $\phi$ is a~correct solution to $I$, then it is also a~correct solution to $\Clear(I; A)$.
\end{observation}

\section{Maximum Weight Independent Set}
In this section we give a data structure that implements the first query, $\QueryMWIS()$, of \cref{thm:main}. 
As mentioned, we will actually solve a~more general problem of \CSP (which we will call \CSPs from now on for brevity).
That is, we will maintain a dynamic instance of \CSPs undergoing the following types of updates:
\begin{itemize}
	\item $\AddVertex(u, D_u, \rev_u)$: adds an isolated~vertex $u$ together with the domain $D_u$ and the revenue function $\rev_u\colon D_u\to \R_{\geq 0}$;
	\item $\AddEdge(u, v, C_{uv})$: adds an edge $uv$ together with the constraint $C_{uv} \subseteq D_u \times D_v$;
	\item $\RemoveEdge(u, v)$: removes an edge $uv$;
	\item $\UpdateRevenue(u, \rev_u)$: changes the revenue function of $u$ to $\rev_u\colon D_u \to \R_{\geq 0}$.
\end{itemize}
After each update, the data structure should maintain an~approximate optimum revenue to the currently stored the instance $I$: i.e., a~nonnegative real $p$ satisfying $(1 - \varepsilon)\OPT \leq p \leq \OPT$, where $\OPT$ denotes the optimum revenue for $I$.
For technical reasons, we do not support removing vertices from the instance; however, observe that the removal of a~vertex can be emulated by removing all edges incident to the vertex and replacing its revenue with the constant-zero function.

Note that the original dynamic instance of \CSPs constructed from {\sc Max Weight Independent Set} has a~fixed $n$-element vertex set, so it will never be updated by $\AddVertex$.
However, the designed data structure will create auxiliary dynamic instances of \CSPs that \emph{will} be maintained using all four types~of~updates.

Henceforth, we assume that a~given apex-minor-free class of graphs $\Cc$ is fixed; that is, we consider all constants depending only on $\Cc$ to be absolute constants.
Moreover, throughout this section we denote by $\eps > 0$ the parameter $\eps$ fixed in the initialization of the data structure.

\subsection{Introductory results}
\label{ssec:indset-static-version}


We begin by summarizing the Baker's technique applied to \CSPs.
In short, we show that if the Gaifman graph has bounded treewidth, then we can solve \CSPs exactly using a standard dynamic programming over the tree decomposition. On the other hand, if the Gaifman graph belongs to $\Cc$, then we can solve the instance approximately by a classical application of the Baker's technique~\cite{Baker94}.

First, the following lemma is obtained by dynamic programming on a tree decomposition.

\begin{lemma}\label{lem:csp-bd-tw}
	Let $I = (G, D, \rev, C)$ be an instance of \CSPs on $n$ vertices and $\Delta \ge 2$ be the maximum size of a domain of any vertex. Then, we can solve $I$ in time $n \cdot \Delta^{\Oh(\tw(G))}$.
\end{lemma}
\begin{proof}
First, we can use an algorithm computing a tree decomposition that is a 2-ap\-prox\-i\-ma\-tion of an optimal one in time $n \cdot 2^{\Oh(\tw(G))}$ by Korhonen \cite{Korhonen21}. Once we have that approximation, we use standard dynamic programming over it that could be seen as a generalization of a similar algorithm solving maximum independent set over graphs with bounded treewidth. For each bag of that decomposition it suffices to have dynamic programming states indexed by all possible valuations of variables from that bag and the transitions follow easily. That dynamic programming takes time $n \cdot \Delta^{\Oh(\tw(G))}$, so as $\Delta \ge 2$, in total this algorithm takes time $n \cdot 2^{\Oh(\tw(G))} + n \cdot \Delta^{\Oh(\tw(G))} = n \cdot \Delta^{\Oh(\tw(G))}$.
\end{proof}
	
Next, the following claim shows how subinstances of small treewidth are constructed by Baker's technique:


\begin{lemma} \label{lem:csp-partitioning}
    Let $G \in \Cc$ be a~connected graph and fix $s \in V(G)$. Then:
    \begin{itemize}
    		\item For integers $\ell \ge 0$, $k \ge 1$, we have $\tw(G[\{v \in V(G) \,\colon\, \ell \le \mathsf{dist}(s, v) < \ell + k\}]) \leq \Oh(k)$;
    		\item For integers $k > i \ge 0$, let $V_i = \{v\,\colon\, \mathsf{dist}(s, v) \equiv i \mod{k}\}$.
    		Then $\tw(G \setminus V_i) \leq \Oh(k)$.
	\end{itemize}        
\end{lemma}
\begin{proof}
  For the first point, let $G \in \Cc$ be connected and fix $s \in V(G)$, $k \ge 1$ and $\ell \ge 0$.
  Let $G'$ be a~minor of $G$ created by:
  \begin{itemize}
    \item discarding from $G'$ all vertices at distance at least $\ell + k$ from $s$;
    \item if $\ell > 0$, contracting all vertices at distance strictly less than $\ell$ to $s$.
  \end{itemize}
  It can be readily verified that the radius of $G'$ is bounded from above by $k$ since $s$ is at distance at most $k$ from every vertex of $G'$.
  Thus by \cref{thm:linearly-locally-bounded-tw}, we have $\tw(G') \leq \Oh(k)$.
  Since $G[A] = G'[A]$, we conclude that $\tw(G[A]) \leq \tw(G')$.
  
  For the second point, choose integers $i, k$ with $0 \le i < k$.
  Observe that $G \setminus V_i$ is a~disjoint union of the subgraphs induced by at most $k-1$ consecutive BFS layers (i.e., for every pair of vertices in a~single component, their distance from $s$ in $G$ differs by at most $k - 1$).
  Hence, the first point of the lemma applies, implying that each of these subgraphs has treewidth at most $\Oh(k)$.
  Therefore, $\tw(G \setminus V_i) \leq \Oh(k)$.
\end{proof}

With these helper lemmas in hand, we can prove that:

\begin{lemma}\label{lem:csp-baker}
Let $\eps' > 0$ be a positive real number and $I = (G, D, \rev, C)$, $G \in \Cc$, be an instance of \CSPs with $n$ vertices, where the maximum size of a domain is $\Delta \ge 2$. One can compute a real $p$ such that $(1 - \eps') \OPT \le p \le \OPT$ in time $n \cdot \Delta^{\Oh\left(1/\eps'\right)}$, where $\OPT$ is the maximum revenue of a solution to $I$.
\end{lemma}
\begin{proof}[Proof of \cref{lem:csp-baker}]
We solve \CSP by a~classical use of the Baker's technique. Let $k = \ceil{\frac{1}{\eps'}}$. We assume that $G$ is connected (because we can solve each connected component separately), choose an arbitrary root $r$, partition the graph into the BFS layers and group them into groups mod $k$. That is, let $V_i = \{v : \mathsf{dist}_G(r, v) \equiv i \mod k\}$ for $i=0, 1, \ldots, k-1$.
Then let $I_i = I \setminus V_i$.
By \cref{lem:csp-partitioning}, the Gaifman graph of each $I_i$ has treewidth at most $\Oh(k)$.
Now, using \cref{lem:csp-bd-tw}, we compute an optimum solution to each $I_i$; call it $\phi_i$ and set $p_i = \rev(\phi_i)$. We claim that it suffices to set $p \coloneqq \mathrm{max}_{i} p_i$.

 For each $i \in \{0, \dots, k - 1\}$, we extend $\phi_i$ to a~valid solution $\phi_i'$ of $I$ by placing value $0$ on each variable from $V_i$. If $\phi$ is an optimum solution to $I$, so that $\OPT = \rev(\phi)$, then $\phi|_{V(G) \setminus V_i}$ is a valid solution to $I_i$, hence $\rev(\phi|_{V(G) \setminus V_i}) \le \rev(\phi_i)$. As $V_0, V_1, \ldots, V_{k-1}$ form a partitioning of $V(G)$, we have that $\rev(\phi|_{V_0}) + \ldots + \rev(\phi|_{V_{k-1}}) = \rev(\phi)$, so there exists some $i$ such that $\rev(\phi|_{V_i}) \le \frac{\rev(\phi)}{k} \le \eps' \OPT$. Therefore, we have that $p \ge \rev(\phi_i) \ge \rev(\phi|_{V(G) \setminus V_i}) = \rev(\phi) - \rev(\phi|_{V_i}) \ge (1-\eps') \OPT$. As each $\phi_i'$ is a valid solution to $I$ and we have that $\rev(\phi_i) = \rev(\phi_i')$, we obviously have that $p \le \OPT$, which completes the proof of the claim.
 
 Solving each $I_i$ requires time $n \cdot \Delta^{\Oh(k)}$. As we solve $k$ auxiliary instances, the total time complexity is $n \cdot k \cdot \Delta^{\Oh(k)} = n \cdot \Delta^{\Oh\left(1/\eps'\right)}$.
\end{proof}

We also present a~useful technical tool allowing us to understand graphs of small treewidth through elimination forests.
The proof follows closely the folklore argument demonstrating, for every graph~$G$, the inequality ${\td(G) \le (\tw(G) + 1) \cdot \log n}$, but we have to additionally ensure the satisfaction of some structural properties of the resulting elimination forest.

\begin{lemma}\label{lem:td-decomp}
	Let $G$ be a graph on $n$ vertices together with a~tree decomposition of width $w$. Then, in time $\Oh(wn \log n)$, one can compute an elimination forest $F$ of $G$ with the following properties:
	\begin{enumerate}
		\item $F$ has height $\Oh(w \cdot \log n)$;
		\item for each $u \in V(F)$, we have $|\Reach_F(u)| = \Oh(w)$;
		\item \label{item:td-decomp-connected} for each $u \in V(F)$, the graph $G[\desc_F[u]]$ is connected.
	\end{enumerate}
\end{lemma}

The proof of \cref{lem:td-decomp} involves a well-known lemma by Bodlaender and Hagerup:
\begin{lemma}[\cite{bodlaender-hagerup}]\label{lem:Hagerup}
	Let $G$ be a graph on $n$ vertices given on input with a~tree decomposition of width at most $w$. Then there exists a tree decomposition of $G$ of width $3w+2$ and height $\Oh(\log n)$.
	Moreover, such a~decomposition can be computed in time $\Oh(wn)$.
\end{lemma}

We now prove \cref{lem:td-decomp}.

\begin{proof}[Proof of \cref{lem:td-decomp}]
	Recall we are given a~tree decomposition, say $(T_0, \bag_0)$, of $G$, of width $w$.
	Using \cref{lem:Hagerup}, we turn it into a~tree decomposition $(T, \bag)$ of $G$ of width at most $3w+2$ and height $\Oh(\log n)$, and root it at an arbitrary node $r$. Now, let us turn this tree decomposition into an elimination forest by ``straightening'' each bag. Let $\gamma : V(T) \to 2^{V(G)}$ be the function that given a~node $t$ of $T$, returns the subset of vertices $v \in \bag(t)$ such that $v \not\in \bag(t')$ for any strict ancestor $t'$ of $t$. In other words, $v \in \gamma(t)$ if and only if $t$ is the shallowest node whose bag contains $v$.
		
	 Now, let us define an~elimination forest $F$ of $G$ by defining the parent relation. Let us order the vertices of each $\gamma(t)$ arbitrarily; say that $\gamma(t)$ is ordered $v^t_1, v^t_2, \ldots, v^t_{|\gamma(t)|}$. For $2 \le i \le |\gamma(t)|$ we let $v^t_{i-1}$ be the parent of $v^t_{i}$ in $F$. If $t$ is the root of $T$, then we let $v^t_1$ to be the root of $F$. Otherwise, we set the parent of $v^t_1$ to $v^{t'}_{|\gamma(t')|}$, where $t'$ is the closest ancestor of $t$ such that $\gamma(t')$ is nonempty. One can readily check that it is in fact a valid elimination forest and that its height is no larger than $\Oh(w \cdot \log n)$.
	 
	As for the property that $|\Reach_F(u)| = \Oh(w)$, observe that if $u \in \gamma(t)$ for some $t \in V(T)$ (as noted before, $t$ is determined uniquely), then $\desc_F[u] \subseteq \gamma(\desc_T[t])$, i.e., for every $v \in \desc_F[u]$, $t$ is an ancestor of all nodes whose bags contain $v$.
	Therefore, $N_G[\desc_F[u]] \subseteq \bag(\desc_T[t])$.
	By the construction of $F$ we have that $\bag(\desc_T[t]) - \bag(t) \subseteq \desc_F[u]$, which implies $\Reach_F(u) = N_G(\desc_F[u]) \subseteq \bag(t)$.
	We conclude that $|\Reach_F(u)| \le 3w+3$, as required.

	However, such $F$ does not have to necessarily fulfill the third condition. Fortunately, the problem is easy to fix in a standard way while maintaining the satisfaction of properties $1$ and $2$. We create a~forest $F'$, where $\parent_{F'}(u)$ is defined as the deepest vertex of $\Reach_F(u)$ in $F$. Note that $\Reach_F(u) \subseteq \anc_F[u]$ and $\anc_F(u)$ forms a path in $F$, so such vertex is well defined; as a~special case, if $\Reach_F(u)$ is empty, then we let $u$ be a root of $F$. One can readily check that $F'$ defined in this way is a valid elimination forest of $G$, still satisfies properties $1$ and $2$, and additionally satisfies that for each $u \in V(F')$ we have that $G[\desc_{F'}[u]]$ is connected.
	
	Moreover, as the proof above is algorithmic, we can compute the required elimination forest $F'$ in time $\Oh(n \log n \cdot w)$.
\end{proof}

The following observation shows how property (\ref{item:td-decomp-connected}) of \cref{lem:td-decomp} can be exploited in algorithms working with elimination forests.

\begin{lemma} \label{cl:diff-reaches}
  Assume $F$ is an~elimination forest of $G$ such that for each $u \in V(F)$, the graph $G[\desc_F[u]]$ is connected.
  If $u, v \in V(F)$ are non-root vertices of $F$ and $\parent_F(u) \neq \parent_F(v)$, then it holds that $\Reach_F(u) \neq \Reach_F(v)$.
\end{lemma}
\begin{proof}
	It is always the case that $\parent_F(u) \in \Reach_F(u)$ as otherwise $\desc_F[\parent_F(u)]$ would not be connected. Also, $\Reach_F(u) \subseteq \anc_F[\parent_F(u)]$, hence $\parent_F(u)$ is the (unique) deepest vertex of $\Reach_F(u)$.
	This finishes the proof.
\end{proof}

\subsection{Compressed instances}

We now define the crucial notion of a \emph{compressed instance} of \CSPs. Intuitively, compression of instances provides us with a~means of producing significantly smaller instances of \CSPs with the same maximum revenue.
We begin with the~definition of a~compressed graph:

\begin{definition}
	\label{def:compressed-graph}
	Let $G$ be a graph and $Y$ be a subset of $V(H)$. By $G\{Y\}$ we denote the \emph{compressed graph} for $G$ and $Y$, defined as follows:
	\begin{align*}
	V(G\{Y\}) &= Y \cup \{N(C) \,\colon\, C \in \cc(G \setminus Y)\} \subseteq Y \cup 2^Y, \\
	E(G\{Y\}) &= E(G[Y]) \cup \{vS: v \in S,\, S \in V(G\{Y\}) \setminus Y\}.
	\end{align*}
	In other words, we group the connected components of $G \setminus Y$ by their neighborhoods and for each such group we create one vertex connected to all vertices from the neighborhood of this group.
\end{definition}

Next we say what it means to compress an instance.

\begin{definition}
	Let $I = (G, D, \rev, C)$ be a \CSPs instance and $Y$ a subset of $V(G)$. We define the compressed instance $I\{Y\} = (G\{Y\}, D', \rev', C')$ as follows.
	\begin{itemize}
	  \item \emph{(Domains.)} If $v \in Y$, then $D'_v = D_v$. If $S \in V(G\{Y\}) \setminus Y$, then $D'_S = \{0\} \cup \prod_{v \in S} D_v$.
	  \item \emph{(Revenues.)} If $v \in Y$, then $\rev'_v = \rev_v$. Now assume $S \in V(G\{Y\}) \setminus Y$. For convenience, let $S = \{u_1, \ldots, u_s\}$ and let $R$ be the union of all connected components of $G \setminus Y$ whose neighborhood is $S$.
	  Let $I_S = I[R \cup S]$.
	  Then $\rev'_S(0) = 0$, and $\rev'_S(d_1, \ldots, d_s)$ for $d_i \in D_{u_i}$ is the maximum revenue of $\rev(\phi|_R)$, where $\phi$ ranges over all valid solutions to $I_S$ such that $\phi(u_i) = d_i$ for each $1 \le i \le s$.
	  If no such solutions exist, then we set $\rev'_S(d_1, \ldots, d_s) = 0$.
	  \item \emph{(Constraints.)} If $u,v \in Y$, then $C'_{uv} = C_{uv}$. If $S \in V(G\{Y\}) \setminus Y$ and $u \in S$, then
	  \[ C'_{uS} = \{(d, e) \,\colon\, d \in D'_u\textrm{ and } e \in D'_S,\, d = 0 \vee e = 0 \vee e(u) = d\}. \]
	  In other words, the constraint $C'_{uS}$ permits the valuations $\phi$ of $G\{Y\}$ for which $\phi(u) = 0$ or $\phi(S) = 0$, or for which $\phi(S) \in \prod_{v \in S} D_v$ and $\phi(S)(u) = \phi(u)$.
	\end{itemize}
\end{definition}

It can be easily verified that for every $I$ and every $Y$, $I\{Y\}$ is a~correct instance of \CSP. Also, note that in the case of defining $\rev'_S(d_1, \ldots, d_s)$, no feasible $\phi$ exists if and only if setting $\phi(u_i) = d_i$ for each $1 \le i \le s$ already violates some constraint, as we can always put $\phi$ as a~zero~function~on~$R$.

In other words, given a~\CSPs instance $I$ and a~set $Y \subseteq V(I)$, we define the compression $I\{Y\}$ by contracting the connected components of $G \setminus Y$ to single vertices, and identifying the contracted vertices with the same neighborhood $S$ in $Y$ to the same vertex $x$.
Then, for each possible valuation $d \in \prod_{v \in S} D_v$ of vertices in $S$, we encode in $x$ the maximum revenue of a~partial solution for the connected components assigned to $x$ that is consistent with $d$.
Note that after a~valuation of vertices in $S$ is fixed, the instances of \CSPs induced by each connected component assigned to $x$ are pairwise independent: there are no constraints between the connected components of $G \setminus Y$.
Therefore:
\begin{observation}
  \label{obs:sum-over-cc}
  Let $I = (G, D, \rev, C)$ be a~\CSPs instance and let $Y \subseteq V(G)$.
  Let $S \in V(G\{Y\}) \setminus Y$, $S = \{u_1, \dots, u_s\}$, and $C_1, C_2, \dots, C_\ell$ be the connected components of $G \setminus Y$ whose neighborhood is $S$.
  For each $1 \le i \le \ell$, let $I_i = I[C_i \cup S]$.
  Let also $d_1 \in D_{u_1}, \dots, d_s \in D_{u_s}$.
  Then, $\rev'_S(d_1, \dots, d_s) = \sum_{i=1}^{\ell} \rev(\phi_i|_{C_i})$, where $\phi_i$ is a solution to $I_i$ maximizing $\rev(\phi_i|_{C_i})$ under the assumption that $\phi_i(u_j) = d_j$ for all $1 \le j \le s$, or $\rev'_S(d_1, \dots, d_s) = 0$ if setting $\phi(u_j) = d_j$ for all $1 \le j \le s$ violates some constraint.
\end{observation}


Having defined compressed graphs and instances, we now present a series of their properties. Namely: $\Cc$ is closed under graph compressions and compressed instances have the same optimum revenues as the original instances.

\begin{lemma}
  \label{lem:compression-stays-in-class}
	Let $G$ be a graph and $Y$ be a subset of $V(H)$. Then $G \in \Cc$ implies $G\{Y\} \in \Cc$.
\end{lemma}
\begin{proof}
	$H\{Y\}$ can be obtained from $H$ through a series of contractions (where we contract each connected component of $H \setminus Y$ into a single vertex) and vertex removals (we remove all but one vertices from each group that has a particular neighborhood in $Y$). Hence, $H\{Y\}$ is a minor of $H$. Since $\Cc$ is assumed to be minor-closed, the claim follows. 
\end{proof}

\begin{lemma}
  \label{lem:compression-preserves-revenue}
	Let $I = (G, D, \rev, C)$ be a \CSPs instance, $Y \subseteq V(G)$ and $I\{Y\} = (G\{Y\}, D', \rev', C')$ be the compressed instance. Then, the optimum solutions of $I$ and $I\{Y\}$ have equal revenues.
\end{lemma}
\begin{proof}
	Let $\OPT$ and $\OPT'$ denote the revenues of optimum solutions to $I$ and $I\{Y\}$.
	
	Let $\phi$ be any solution to $I$. We shall define $\phi'$ which will be a~valid solution to $I\{Y\}$. We define $\phi'$ to agree with $\phi$ on $Y$.
	Then $\rev'(\phi'|_Y) = \rev(\phi|_Y)$.
	Now choose $S \in \{N(E) \,\colon\, E \in \cc(H \setminus Y)\}$. Let $S = \{u_1, \ldots, u_s\}$ and $\phi(u_1)=d_1, \ldots, \phi(u_s) = d_s$. We set $\phi'(S) \coloneqq (d_1, \ldots, d_s)$.	
	It can be readily verified that $\phi'$ complies with all the constraints in $C'$, so $\phi'$ is valid.
	Now, if $R$ is the union of all the connected components of $H \setminus Y$ whose neighborhood is equal to $S$, then $\rev'(\phi'(S)) \ge \rev(\phi|_R)$ from the definition of $\rev'$.
	By adding such inequalities for all possible pairs of $S$ and $R$ and the equality $\rev'(\phi'|_Y) = \rev(\phi|_Y)$, we conclude that $\rev'(\phi') \ge \rev(\phi)$, hence $\OPT' \ge \OPT$.
	
	On the other hand, let $\psi'$ be any solution to $I\{Y\}$. We shall define $\psi$ -- a valid solution to~$I$. We define $\psi$ to agree with $\psi'$ on $Y$, hence $\rev(\psi|_Y) = \rev'(\psi'|_Y)$.
	Now choose $S \in V(H\{Y\}) \setminus Y$ and assume that $S = \{u_1, \ldots, u_s\}$. 
	Let $R$ be the union of all the connected components of $H \setminus Y$ whose neighborhood is equal to $S$.
	If $\rev'(\psi'(S)) = 0$, then we set $\psi(v) \coloneqq 0$ for all $v \in R$.
	Otherwise, $\rev'(\psi'(S)) > 0$, so $\psi'(S) \in \prod_{v \in S} D_v$.
	In this case, from the definition of $C'_{u_i S}$, it has to be that $\psi'(S)(u_1) = \psi'(u_1) = \psi(u_1), \ldots, \psi'(S)(u_s) = \psi'(u_s) = \psi(u_s)$, and let $I_S = I[R \cup S]$.
	Since $\rev'(\psi'(S)) > 0$, there exists a~solution $\zeta$ to $I_S$ such that $\zeta(u_i) = \psi'(u_i) = \psi(u_i)$ for all $1 \le i \le s$.
	Let $\zeta$ be any such maximum-revenue solution, and set $\psi|_R \coloneqq \zeta|_R$.
	Then $\rev(\psi|_R) = \rev'(\psi'(S))$ by the definition of $\rev'$.
	It is also easy to verify that $\psi$ is valid.
	It follows that $\rev(\psi) = \rev'(\psi')$, so $\OPT \ge \OPT'$.
\end{proof}

\subsection{Dynamic maintenance of compressed instances}
\label{ssec:indset-dynamic-compressions}
We proceed to showing how the notion of compressed instances can be used in the setting of the dynamic maintenance of 2CSPs.
Consider an instance $I = (G, D, \rev, C)$ of \CSPs, where $G \in \Cc$ and $\tw(G) \leq w$ for some $w \in \N$.
The instance $I$ undergoes a~sequence of updates, so that at every point of time we have that $G \in \Cc$ (but it might not necessarily be the case that $\tw(G) \leq w$).
Our goal is to maintain a~much smaller \CSPs instance $I^\star$ with the same value of optimum solution as $I$, so that at all points of time, the size of $I^\star$ is roughly proportional to the number of updates $I$ has undergone so far.
This is formalized by the following lemma.

\begin{lemma}
  \label{lem:indset-dynamic-compression}
  Let $w, n, \Delta \in \N$, $\Delta \ge 2$.
  There is a~data structure that supports the following operations:
  \begin{itemize}
    \item \textbf{Initialize} the data structure with an~$n$-vertex \CSPs instance $I = (G, D, \rev, C)$, where $G \in \Cc$, $\tw(G) \leq w$, and all vertices in $I$ have domains not larger than $\Delta$.
    \item \textbf{Update} the instance $I$ using one of the following update types: $\AddVertex$, $\AddEdge$, $\RemoveEdge$, $\UpdateRevenue$.
    It is guaranteed that after the update, we have that $G \in \Cc$ and all vertices in $I$ have domains not larger than $\Delta$.
  \end{itemize}
  The initialization of the data structure is performed in time $\Delta^{\Oh(w)} \cdot n \log n$.
  Afterwards, the data structure additionally maintains a~\CSPs instance $I^\star = (G^\star, D^\star, \rev^\star, C^\star)$, updated by $\AddVertex$, $\AddEdge$, $\RemoveEdge$ and $\UpdateRevenue$, with the following properties:
  \begin{enumerate}[label=(\alph*)]
    \item \label{item:compressreq-in-class} $G^\star \in \Cc$;
    \item \label{item:compressreq-revenue} the maximum revenue of a~solution to $I^\star$ is equal to that of $I$;
    \item \label{item:compressreq-domain} each vertex of $I^\star$ has domain bounded in size by $\Delta^{\Oh(w)}$;
    \item \label{item:compressreq-size} after a~sequence of $t \geq 0$ updates to $I$, we have $|V(I^\star)| \leq t \cdot w^{\Oh(1)} \log n$.
    Moreover, on each update to $I$, the instance $I^\star$ can be updated in time $\Delta^{\Oh(w)} \log n$ and this causes at most $w^{\Oh(1)} \log n$ updates to $I^\star$.
  \end{enumerate}
\end{lemma}

\begin{proof}
  On initialization, we compute a~tree decomposition of $G$ of width at most $\Oh(w)$ in time \mbox{$2^{\Oh(w)} \cdot n$}, for instance using the fixed-parameter algorithm of Korhonen~\cite{Korhonen21}.
  We then transform this tree decomposition into an elimination forest $F$ of $G$ with $V(F) = V(G)$ such that: (i) $F$ has height $\Oh(w \log n)$; (ii) for each $u \in V(F)$, we have $|\Reach_F(u)| = \Oh(w)$; and (iii) for each $u \in V(F)$, the graph $G[\desc_F[u]]$ is connected.
  This can be done in time $\Oh(w \cdot n \log n)$ by \cref{lem:td-decomp}.
  
  Moreover, we precompute multiple dynamic programming tables on $F$:
  \begin{itemize}
    \item For each $u \in V(F)$ and for each possible valuation $\phi \in \prod_{s \in \Reach_F(u)} D_s$ of vertices from $\Reach_F(u)$, compute $T[u][\phi]$: the maximum possible revenue $\rev(\psi|_{\desc_F[u]})$ over all solutions $\psi$ to $I[\desc_F[u] \cup \Reach_F(u)]$ agreeing with $\phi$ on $\Reach_F(u)$, or $0$ if no such solution exists.
    Note that $|\Reach_F(u)| = \Oh(w)$, hence there are at most $\Delta^{\Oh(w)}$ possible valuations of $\phi$ on $\Reach_F(u)$ for a given $u$.
    Also note that for a vertex $u$ with children $c_1, \ldots, c_s$ we have that $\Reach_F(c_i) \subseteq \Reach_F(u) \cup \{u\}$.
    Hence, we can compute all values $T[\cdot][\cdot]$ using a simple bottom-up dynamic programming in time~${n \cdot \Delta^{\Oh(w)}}$.
    
    \item For each $u \in V(F)$, compute the set $\Nc_u = \{\Reach_F(c) \,\colon\, c\text{ is a child of }u\text{ in F}\}$.
    This can be performed easily in time $\Oh(nw)$.
    Note that $R \subseteq \Reach_F(u) \cup \{u\}$ for all $R \in \Nc_u$.
    
    \item For each $u \in V(F)$, $R \in \Nc_u$ and valuation $\phi  \in \prod_{s \in R} D_s$ of vertices from~$R$, compute $W[u][R][\phi]$: the sum of $T[c][\phi]$ such that $c$ is a~child of $u$ and $\Reach_F(c) = R$.
    This table can be computed in time $n \cdot \Delta^{\Oh(w)}$ as well: for each non-root vertex $c$ and each valuation $\phi$ of $\Reach_F(c)$, the value $T[c][\phi]$ is added to exactly one entry of~$W$.
  \end{itemize}
  
  We remark that the dynamic programming tables above suffice to compute the maximum-revenue solution for $I$ at the time of the initialization: it is simply the sum over $T[r][\emptyset]$ ranging over all roots $r$ of trees of the elimination forest $F$.
  
  Also initialize sets $A = B = Z = \emptyset$ and $I^\star = I\{\emptyset\}$.
  We remark that $I\{\emptyset\}$ is a~one-vertex instance of 2CSP with the same revenue of the optimum solution as $I$.
  
  Let $\Iinit = (\Ginit, D^{\mathrm{init}}, \rev^{\mathrm{init}}, C^{\mathrm{init}})$ be the initial instance of \CSPs. In the sequel, by $\Icur = (\Gcur, D^{\mathrm{cur}}, \rev^{\mathrm{cur}}, C^{\mathrm{cur}})$ we denote the current snapshot of the instance  in the data structure; initially, $\Icur \coloneqq \Iinit$.
  Note that $V(\Iinit) \subseteq V(\Icur)$ at all times since vertices cannot be explicitly removed from $I$.
  Instead, we emulate the removals of vertices from \CSPs by removing all the edges incident to the vertex and setting the revenue of the vertex to the zero function. If two instances of \CSPs are isomorphic after removing from both of them the sets of isolated vertices with zero revenues, we say that these instances are \emph{equivalent}. 
  
  Throughout the life of the data structure, we maintain the following invariants:
  \begin{itemize}
    \item $A = V(\Icur) \setminus V(\Iinit)$ is the set of vertices added to $I$ since the instantiation of the structure.
    \item $B \subseteq V(\Icur)$ is the set of vertices that were part of any update to $I$ so far (i.e., $v \in B$ if $v$ was added to $I$, the revenue of $v$ was changed, or an~edge incident to $v$ was added or removed).
    \item $Z = A \cup \anc_F[B \setminus A]$, i.e., $Z$ contains $A$ and all the ancestors in $F$ of all vertices of $B \setminus A$.
    \item $I^\star$ is equivalent to $\Icur\{Z\}$.
  \end{itemize}
  Note that it follows that $A \subseteq B$, $V(\Icur) = V(\Iinit) \cup A$ and $B \setminus A \subseteq V(\Iinit)$. Also, the invariants are satisfied at the time of initialization.
  Moreover, the required properties \ref{item:compressreq-in-class}, \ref{item:compressreq-revenue} and \ref{item:compressreq-domain} follow from the invariant: since $I^\star$ is equivalent to $\Icur\{Z\}$, property \ref{item:compressreq-in-class} follows from \cref{lem:compression-stays-in-class}, property \ref{item:compressreq-revenue} follows from \cref{lem:compression-preserves-revenue}, and property \ref{item:compressreq-domain} follows from the definition of a~compressed instance and the fact that the domain of each vertex of $I$ has size bounded by $\Delta$.
  
  We envision $Z$ as consisting of two parts --- an ``unordered cloud'' of newly added vertices (which is $A$) and a well-structured part that is the smallest prefix of $F$ containing all vertices from $V(\Iinit)$ that were part of any update. Note that for any vertex $v \notin Z$, we have $v \notin B$. Therefore, any $v \notin Z$ is not adjacent to any vertex of $A$, and is adjacent to some vertex of $\Icur$ if and only if it is adjacent to this vertex in $\Iinit$.
  In other words, the neighborhood of $v$ has stayed unchanged; that is, $v \notin Z \Rightarrow N_{\Gcur}(v) = N_{\Ginit}(v)$. The main idea of our approach is that we treat $Z$ as a possibly complicated part, whereas the behavior of the vertices outside of $Z$ can be well-understood through the preprocessed dynamic programming tables.
  Thus, we can effectively compress the parts of the graph outside of $Z$ and construct a~concise instance of \CSPs without changing the revenue of the optimum solution.
  We will prove that the size of the compressed instance depends mainly on the size of $Z$ -- intuitively, we do not compress $Z$ in any way, but the parts of the graph outside of $Z$ will be compressed heavily. Fortunately, the size of $Z$ cannot increase too much with a single update:

\begin{claim} \label{cl:growing-z}
	The set $Z$ expands by at most $\Oh(w \log n)$ vertices with each update to $I$.
\end{claim}
\begin{claimproof}
	Naturally, all sets $A$, $B$, $Z$ can only expand after each update.
	
	If the update to $I$ is of type $\AddVertex(u,D_u,\rev_u)$, then sets $A$ and $B$ grow by one vertex $u$, and $\anc_F[B \setminus A]$ does not change. Hence $Z$ grows by exactly one vertex.
	
	For our convenience, let $P_u$ for a vertex $u \in V(\Gcur)$ be the empty set if $u \in A$ and $\anc_F[u]$ otherwise.
	If the update is of type $\AddEdge(u, v, c)$ or $\RemoveEdge(u, v)$, then $A$ does not change, $B$ expands by $u$ and $v$ (if these vertices were not part of $B$ yet), and $\anc_F[B \setminus A]$ grows by vertices that are contained within $P_u \cup P_v$. Since $F$ is of height $\Oh(w \log n)$, it follows that $Z$ grows by $\Oh(w \log n)$ vertices.
	
	Finally, if the update is of type $\UpdateRevenue(v, r_v)$, then $A$ does not change, $B$ additionally includes $v$ (if not in $B$ yet), and $\anc_F[B \setminus A]$ grows by a set of vertices that is contained within $P_v$, which again is of size $\Oh(w \log n)$.
\end{claimproof}

  Let us now understand the structure of the compressed instance $\Icur\{Z\}$.
  Recall that for each $u \in V(F)$, it holds that $\desc_F[u]$ induces a connected graph in $\Ginit$. Therefore, if $u \notin Z$, then $\desc_F[u]$ induces a connected subgraph of $\Gcur$ as well. If additionally $\parent_F(u) \in Z$, then $N_{\Gcur}(\desc_F[u]) = N_{\Ginit}(\desc_F[u]) = \Reach_F(u) \subseteq \anc_F(u) \setminus \{u\} \subseteq Z$, so we conclude that $\desc_F[u]$ is actually a~connected component of $\Gcur \setminus Z$.
   Hence, there exists a natural bijection between the connected components of $\Gcur \setminus Z$ and the vertices $u$ such that $u \notin Z$, but $\parent_F(u) \in Z$ (or $u \notin Z$ is the root of some tree in $F$).
   We will call all such vertices $u$ \emph{appendices} of $Z$ in $F$.

   Observe that the revenues of the vertices in the compressed instance can be inferred from the entries of the precomputed dynamic programming table $T[\cdot][\cdot]$.
   In the instance $\Icur\{Z\}$, consider $S \in V(G\{Z\}) \setminus Z$ and let $\phi \in \prod_{s \in S} D_s$ be a~valuation of vertices in $S$.
   Then, $\rev^{\mathrm{cur}}(S, \phi)$ is by \cref{obs:sum-over-cc} exactly the sum of $T[u][\phi]$ over all appendices $u$ of $Z$ in $F$ such that $\Reach_F(u) = S$.

It turns out that it is possible to maintain $\Gcur\{Z\}$ efficiently given the precomputed tables $T$ and $W$.
To this end, we claim that $G\{Z\}$ can be updated efficiently when an~appendix of $Z$ is added to $Z$.


\begin{claim} \label{cl:add-vtx-compr}
	Let two sets $Z_1, Z_2 \subseteq V(\Gcur)$ be such that $A \subseteq Z_1 \subseteq Z_2 \subseteq V(\Gcur)$, $|Z_2| = |Z_1| + 1$ and $Z_1 \setminus A$ and $Z_2 \setminus A$ are prefixes of $F$ with $Z_2 - Z_1 \subseteq V(F)$. Then, an instance equivalent to $\Icur\{Z_2\}$ can be obtained from an instance equivalent to $\Icur\{Z_1\}$ through a sequence of $w^{\Oh(1)}$ updates.
	Moreover, this sequence can be computed in time $\Delta^{\Oh(w)}$.
\end{claim}
\begin{claimproof}
	Let $z \in V(F)$ be such that $Z_2 = Z_1 \cup \{z\}$.
	Let $\Cfrak_1$ be the set of connected components of $\Gcur \setminus Z_1$, $\Cfrak_2$ be the set of connected components of $\Gcur \setminus Z_2$ and $c_1, \ldots, c_t$ be the set of children of $z$ in $F$. We have that $\Cfrak_1 \setminus \Cfrak_2 = \{\Gcur[\desc_F[z]]\}$ and $\Cfrak_2 \setminus \Cfrak_1 = \{\Gcur[\desc_F[c_1]], \ldots, \Gcur[\desc_F[c_t]]\}$.
	
	Let $S = \Reach_F(z)$. The compressed instance $\Icur\{Z_1\}$ contains a vertex $v_S$ representing the union of all connected components of $\Gcur \setminus Z_1$ whose neighborhoods are exactly $S$; and $\Gcur[\desc_F[z]]$ is one of such components. To obtain $\Icur\{Z_2\}$ from $\Icur\{Z_1\}$, we need to:
	\begin{enumerate}
	  \item Remove the contribution of $\Gcur[\desc_F[z]] \in \Cfrak_1 \setminus \Cfrak_2$ from the compressed instance.
	  If $\desc_F[z]$ is the only connected component of $\Gcur \setminus Z_1$ with neighborhood $S$, then the vertex $v_S$ should be removed from the instance; this is emulated by removing all edges incident to $v_S$ and replacing the revenue of $v_S$ with the zero function.
	  Otherwise, the revenue of $v_S$ is updated by subtracting, for each valuation $\phi \in \prod_{s \in S} D_s$ of $S$, the entry $T[v][\phi]$ from the revenue of vertex $v_S$ in state $\phi$.
	  In both cases, the compressed vertex $v_S$ has degree $|\Reach_F(z)| \leq \Oh(w)$, so in either case we apply at most $\Oh(w)$ updates to the compressed instance.
	  \item Add the vertex $z$ to the compressed instance.
	  Since the set of neighbors of $z$ in $Z_1$ is exactly $\Reach_F(z)$, this requires one vertex addition and $\Oh(w)$ constraint additions.
	  \item Include the contribution of the connected components $\Gcur[\desc_F[c_1]], \dots, \Gcur[\desc_F[c_t]] \in \Cfrak_2 \setminus \Cfrak_1$ in the compressed instance.
	  Even though $t$ might possibly be large, the number of different neighborhoods of the new connected components is bounded -- this follows from $N_{\Gcur}(\desc_F[c_i]) = \Reach_F(c_i) \subseteq \Reach_F(z) \cup \{z\}$.
	  Since $|\Reach_F(z)| \leq \Oh(w)$, by \cref{cor:neighborhood-complexity-components} there exist at most $\Oh(w)$ distinct neighborhoods of the new components; in other words, $|\Nc_z| \leq \Oh(w)$.
	  By \cref{cl:diff-reaches}, all these neighborhoods actually correspond to new vertices in the compressed instance.
	  Thus, for each $S' \in \Nc_z$ (equivalently, for each $S' \subseteq \Reach_F(z) \cup \{z\}$ with at least one component $\desc_F[c_i]$ with the neighborhood equal to $S'$), we add to the compressed instance a~new vertex $v_{S'}$.
	  For each valuation $\phi \in \prod_{s \in S'} D_s$ of vertices in $S'$, the revenue of $v_{S'}$ in state $\phi$ is the sum over all $T[c_i][\phi]$ for all $1 \le i \le t$ such that $\Reach_F(c_i) = S'$; this is exactly the preprocessed value $W[z][S'][\phi]$.
	  Each of the new $\Oh(w)$ vertices has degree at most $\Oh(w)$, so the number of updates to the compressed instance is $\Oh(w^2)$.
	  For each fresh vertex $v_{S'}$, we iterate over all $\Delta^{\Oh(w)}$ valuations of $S'$ to produce the revenues for each state of $v_{S'}$.
	  This can be done in total time $\Oh(w) \cdot \Delta^{\Oh(w)} = \Delta^{\Oh(w)}$. \hfill\qedhere
	\end{enumerate}
\end{claimproof}

Now the proof of the lemma follows easily from \cref{cl:growing-z} and \cref{cl:add-vtx-compr}. If the update to $\Icur$ is of type $\AddVertex(r_u)$, then $Z$ grows by one isolated vertex and we just pass that update to $I^\star$. Otherwise, the set $A$ does not change and $Z$ grows by $\Oh(w \log n)$ vertices from $F$ that can be easily determined. We process the additions to $Z$ vertex by vertex, applying \cref{cl:add-vtx-compr} on each addition.
Note that the vertices should be added to $Z$ in the order of the increasing distance from the root of $F$, so as to ensure that $Z \setminus A$ remains a~prefix of $F$ at all times.
Finally, after all the required vertices are included in $Z$, we relay the queried update of $\Icur$ to $I^\star$. Since the vertices involved in the update (the vertex whose revenue is changed or both of the endpoints of an~updated edge) are now in $Z$, this update can be passed verbatim to $I^\star$.
It is now easy to verify that all the stated invariants are preserved by the update.

We invoke \cref{cl:add-vtx-compr} at most $\Oh(w \log n)$ times, so the total number of updates performed on $I^\star$ is $(w \log n) \cdot w^{\Oh(1)} = w^{\Oh(1)} \log n$.
Also, the total update time is $\Delta^{\Oh(w)} \log n$.
This satisfies the required property \ref{item:compressreq-size} of the described data structure and concludes the proof.
\end{proof}

\subsection{Full algorithm}
We now describe the implementation of the data structure in detail.
%
Recall that our aim is to maintain an~$n$-vertex dynamic instance $I^{\mathrm{main}} = (G^{\mathrm{main}}, D^{\mathrm{main}}, \rev^{\mathrm{main}}, C^{\mathrm{main}})$ of \CSPs in a~data structure, updated by $\AddEdge$, $\RemoveEdge$, and $\UpdateRevenue$, that can be queried for an~approximate optimum revenue in the current snapshot of the instance: for a~parameter $\eps > 0$ fixed at the initialization, the data structure should, when queried, return a~nonnegative real $p$ such that $(1 - \eps)\OPT \le p \le \OPT$.
The initialization of the data structure should take time $f(\eps) \cdot n^{1 + o(1)}$ and each update should take amortized time $f(\eps) \cdot n^{o(1)}$, for some function $f$ that is doubly-exponential in $\Oh(1/\eps^2)$.

We fix an~integer $L \in \N$, whose value will be determined later, and set $k \coloneqq \left\lceil L/\eps \right\rceil$.
We construct a~recursive, $L$-level data structure. That is, we will maintain a~collection of auxiliary data structures, each maintaining an~instance of \CSPs and an~approximate optimum revenue to the maintained instance. Each auxiliary data structure will be assigned to one of the levels $1, \dots, L$.
At level $L$, we have a~single auxiliary data structure $\D^{\mathrm{main}}$ maintaining $I^{\mathrm{main}}$. Next, consider an~auxiliary data structure $\D$ at level $q \in \{1, \dots, L\}$, and assume $\D$ maintains an~instance $I = (G, D, \rev, C)$ of \CSPs.
If $q \ge 2$, then $\D$ maintains a~collection of $k$ data structures $\D_0, \dots, \D_{k-1}$ at level $q - 1$, called \emph{children} of $\D$, with each child maintaining an~instance derived from $I$.
If $q \le L - 1$, then $\D$ is maintained by exactly one data structure at level $q + 1$, called the \emph{parent} of $\D$.
Note that this way, the entire collection of auxiliary data structures forms a~rooted tree of height $L$ and branching $k$, where the root is $\D^{\mathrm{main}}$ and the leaves are the data structures at level $1$. In particular, the total number of maintained data structures is $\Oh(k^L)$.

Moreover, each auxiliary data structure $\D$ at level $q$ preserves the following invariants:

\begin{enumerate}[label=(I\arabic*)]
  \item \label{item:indset-struct-size} $G \in \Cc$ and $|V(G)| \leq n^{q / L}$;
  \item \label{item:indset-struct-domain} for each $v \in V(G)$, we have $|D_v| \leq g(q)$ for some function $g$ to be specified in \cref{ssec:indset-parameters};
  \item \label{item:indset-struct-sol} $\D$ maintains a~nonnegative real $p$ satisfying $(1 - \eps \cdot \frac{q}{L})\OPT \le p \le \OPT$, where $\OPT$ is the maximum revenue of a~solution to $I$.
\end{enumerate}

Note that \ref{item:indset-struct-size} and \ref{item:indset-struct-domain} are satisfied by $\D^{\mathrm{main}}$ with $g(L) = 2$. By \ref{item:indset-struct-sol}, upon query $\QueryMWIS()$, the data structure can just return the real $p$ stored in $\D^{\mathrm{main}}$.

Now, we explain the implementation of each data structure.
First, let $\D$ be a~data structure at level $1$ maintaining an~instance $I$ of \CSPs.
On each update to $I$, we update $I$ by brute force: we recompute the approximate solution from scratch using \cref{lem:csp-baker} with $\eps'$ supplied to it equal to $\frac{\eps}{L}$.
Thus, processing each update to $I$ takes time $|V(I)| \cdot g(1)^{\Oh(L/\varepsilon)} \leq n^{1/L} \cdot g(1)^{\Oh(L/\varepsilon)}$.
Thus each level-$1$ data structure satisfies invariant \ref{item:indset-struct-sol}.

From that point on, let us fix some $2 \le q \le L$ and describe the implementation of a~data structure~$\D$ of level $q$, which maintains a \CSP instance $I = (G, D, \rev, C)$, assuming invariants \ref{item:indset-struct-size} and \ref{item:indset-struct-domain}.
Let $\Icur = (\Gcur, D^{\mathrm{cur}}, \rev^{\mathrm{cur}}, C^{\mathrm{cur}})$ denote the current snapshot of $I$.
The lifetime of $\D$ is partitioned into \emph{epochs}: sequences of $\tau_q$ updates to $I$, with $\tau_q$ to be specified later.
The first epoch begins when $\D$ is initialized; and a~new epoch begins each time $\D$ processes $\tau_q$ updates to $I$.
At the start of each epoch, we let $\Iold \coloneqq \Icur$ and we apply the Baker's scheme to $\Iold$.
That is, let $\Iold = (\Gold, D^{\mathrm{old}}, \rev^{\mathrm{old}}, C^{\mathrm{old}})$.
We produce a~partitioning of $V(\Iold)$ into layers $V_0, \dots, V_{k - 1}$ as follows.
Assume that $\Gold$ is connected; otherwise apply the scheme to each connected component of $\Gold$, and let the $i$th layer $V_i$ be the union of the $i$th layers for each connected component of $\Gold$.
Now choose an arbitrary vertex $s$ as a~root, run the breadth-first search (BFS) on $\Gold$ from $s$, partition the graph into the BFS layers, and group them based on modulo $k$, getting sets $V_i = \{v : \mathsf{dist}_{\Gold}(s, v) \equiv i \mod k\}$ for $i=0, 1, \ldots, k-1$.

Now, given sets $V_0, \dots, V_{k-1}$, we define $k$ dynamic instances of \CSPs: $\Icur_i = \Icur \setminus V_i$, called \emph{universes}.
At the start of the epoch, $\Icur_i = I^{\mathrm{old}} \setminus V_i$.
Note that by \cref{lem:csp-partitioning}, we have that $\tw(\Gold_i) \leq \Oh(k)$.
Thus, for each $0 \le i \le k - 1$, we can instantiate a~data structure of \cref{lem:indset-dynamic-compression} on each $\Icur_i$, with $w = \Oh(k) = \Oh(L / \eps)$ and $\Delta = g(q)$.
Since $|V(\Iold)| \leq n^{q/L}$, the initialization of each structure takes time $g(q)^{\Oh(L/\eps)} \cdot n^{q/L} \log n$.
After the initialization, the $i$th data structure maintains a~compressed instance $I^\star_i$ of \CSPs with the same revenue of the optimum solution as $\Icur_i$, each initially containing one vertex.
Therefore, to finish the initialization, we recursively spawn a~collection of $k$ children auxiliary data structures $\D_0, \dots, \D_{k-1}$ at level $q-1$.
Each $\D_i$ stores the compressed instance $I^\star_i$ and maintains a~$(1 - \varepsilon \frac{q-1}{L})$-approximate optimum revenue to $I^\star_i$ (which is also the $(1 - \varepsilon \frac{q-1}{L})$-approximate optimum revenue to $\Icur_i$).
Note that each of the children data structures recursively spawn $k$ additional children data structures at level $q-2$, etc., until constructing $k^{q-1}$ data structures at level $1$ in total.
Each of these data structures is initialized with a one-vertex \CSPs instance.

Now, assuming we can maintain a~good approximation $p_i$ of the maximum-revenue solution to each $\Icur_i$, we can also maintain such a~fairly good approximation to $I$ by simply keeping the maximum $p_i$:

\begin{lemma}
	\label{lem:chaining-approximations}
	Let $\OPT_i$ be the optimum revenue of a solution in the instance $\Icur_i$, and let $p_i$ be such that $(1 - \eps \frac{q-1}{L})\OPT_i \le p_i \le \OPT_i$. Let $p = \max(p_0, \ldots, p_{k-1})$.
	Then $(1- \eps \frac{q}{L})\OPT \le p \le \OPT$, where $\OPT$ is the revenue of the optimum solution to $\Icur$.
\end{lemma}
\begin{proof}
%
	Recall that $\Gold$ is the Gaifman graph of $\Iold$ and $\Gcur$ is the Gaifman graph of $\Icur$.
	Also, let $A$ be the set of vertices that were added to $\Icur$ since the beginning of the current epoch. Then, we have that $V(\Gcur) = V(\Gold) \cup A$. Let $\Gcur_i$ be the Gaifman graph of the universe $\Icur_i$. Then $V_i \subseteq V(\Gcur)$ and $V(\Gcur_i) = V(\Gcur) \setminus V_i$.
	Let $\phi, \phi_0, \ldots, \phi_{k-1}$ be optimum solutions to the instances $\Icur, \Icur_0, \ldots, \Icur_{k-1}$. As $V_0, \ldots, V_{k-1}$ are disjoint, we have that $\rev(\phi|_{V_0}) + \ldots + \rev(\phi|_{V_{k-1}}) \le \rev(\phi)$. Therefore, there exists $0 \le i \le k-1$ such that $\rev(\phi|_{V_i}) \le \frac{\rev(\phi)}{k} \le \rev(\phi) \cdot \frac{\eps}{L}$.

	As $\phi|_{V(\Gcur_i)}$ is a valid solution to $\Icur_i$, we have that $\rev(\phi|_{V(\Gcur_i)}) \le \rev(\phi_i) = \OPT_i$. In turn, we have $(1-\frac{\eps}{L}) \rev(\phi) \le \rev(\phi) - \rev(\phi|_{V_i}) = \rev(\phi|_{V(\Gcur) \setminus V_i}) = \rev(\phi|_{V(\Gcur_i)}) \le \OPT_i$. By multiplying both sides by $1-\eps \frac{q-1}{L}$, we get $(1-\frac{\eps}{L})(1-\eps \frac{q-1}{L}) \OPT \le (1-\eps \frac{q-1}{L}) \OPT_i \le p_i \le p$, which implies that $p \ge \OPT (1-\frac{\eps}{L})(1-\eps \frac{q-1}{L}) \ge \OPT (1-\eps \frac{q}{L})$, as required. On the other hand, as any $\phi_j$ can be extended to the solution of the full instance by putting zeros on $V_j$, we obviously get $p_j \le \OPT_j \le \OPT$ for each $j=0, \ldots, k-1$, which in turn implies that $p \le \OPT$, as desired.
\end{proof}

Each update to $\Icur$ is processed by $\D$ as follows: for each $i \in \{0, \dots, k-1\}$, if the update involves a~vertex of $V_i$, then $\Icur_i$ remains unchanged by the update and no further action is required.
Otherwise, the update is relayed to $\Icur_i$.
The data structure from \cref{lem:indset-dynamic-compression} produces a~sequence of at most $k^{\Oh(1)} \log n$ updates to $I^\star_i$, which are then relayed to $\D_i$.
After all children data structures process the updates, $\D$ recomputes its approximate solution to $\Icur$ by querying each $\D_i$ for the approximation of the maximum revenue of a~solution to $I^\star_i$ and returning the maximum value.
Note that during one epoch, the size of each $I^\star_i$ does not grow above $\tau_q \cdot k^{\Oh(1)} \log n$ (\cref{lem:indset-dynamic-compression}\ref{item:compressreq-size}).
By choosing $\tau_q$ so that this value is significantly less than $n^{(q-1)/L}$, we ensure the satisfaction of invariant \ref{item:indset-struct-size} by the children data structures (the fact that the Gaifman graph of $I^\star_i$ belongs to $\Cc$ follows from \cref{lem:indset-dynamic-compression}\ref{item:compressreq-in-class}).
The invariant \ref{item:indset-struct-domain} is satisfied due to \cref{lem:indset-dynamic-compression}\ref{item:compressreq-domain}, provided we set $g$ so that $g(q-1) \geq g(q)^{\Oh(k)}$. 
Also, since the revenue of an optimum solution to $I^\star_i$ is equal to that of $\Icur_i$ (\cref{lem:indset-dynamic-compression}\ref{item:compressreq-revenue}), each child data structure $\D_i$ maintains a~nonnegative real $p_i$ satisfying the preconditions of \cref{lem:chaining-approximations}.
Therefore, the value $p \coloneqq \max(p_0, \dots, p_{k-1})$ is an~$(1 - \eps \frac{q}{L})$-approximation of the optimum revenue in $\Icur$, which proves the satisfaction of invariant \ref{item:indset-struct-sol} by~$\D$.

Finally, when an~epoch in $\D$ ends, $\D$ is reinitialized with $\Iold \coloneqq \Icur$ and a~new epoch starts.
We remark that this causes the destruction and the recursive reinitialization of the children data structures~$\D_i$.

\subsection{Setting the parameters and time complexity analysis}
\label{ssec:indset-parameters}
In this section, we are going to discuss the parameters whose specification was postponed: $L$, function $g$, and epoch lengths. Finally, we analyze the amortized time complexity of the updates.

At first, we are going to bound the domain sizes.


\begin{lemma} \label{lem:g-bound}
    One can set function $g$ so that invariant \ref{item:indset-struct-domain} is satisfied for all the constructed data structures and
	$g(q)\in  2^{k^{\Oh(L)}}$ for all $q\in \{1,\ldots,L\}$.
\end{lemma}
\begin{proof}
	Recall that invariant \ref{item:indset-struct-domain} states that the domain sizes in instances stored by data structures at level $q$ are bounded by $g(q)$. For $q=L$ it suffices to set $g(L) = 2$. Next, for an instance within some data structure at level $q>1$ we create some number of universes, and instances created for those universes have the same domains as the original one, say of size at most $\Delta$. However, when compressing an instance and recursing to the deeper level, the domain sizes increase to at most $\Delta^{Ck}$ for some constant $C$. Indeed, recall that for a compressed vertex $S$ we had $D'_S = \{0\} \cup \prod_{v \in S} D_v$, and due to the structure of $F$ and $Z$ we know that such $S$ is always of size $\Oh(k)$. It follows that we may set $g(q-1)=g(q)^{Ck}$.
	A straightforward induction now shows that $g(q) = 2^{(Ck)^{L - q}}$, so in particular $g(q) \le 2^{k^{\Oh(L)}}$ for all~$q\in \{1,\ldots,L\}$.
\end{proof}

We know that each instance at level $q$ spawns $k$ universes and each universe spawns one instance at level $q-1$; unless $q=1$, in which case the instance in question is a leaf. Hence, there are at most $k^L$ instances at level $1$. We can proceed further with this reasoning to show the following.

\begin{lemma} \label{lem:updates-count}
	Each update to $\D^{\mathrm{main}}$ causes at most $(k \log n)^{\Oh(L)}$ updates throughout all data structures.
\end{lemma}
\begin{proof}
Recall that by \cref{lem:indset-dynamic-compression}\ref{item:compressreq-size}, each update to a universe at level $q>1$ causes at most $(k \log n)^{\Oh(1)}$ updates propagated to instances at level $q-1$. As each instance at level $q>1$ has $k$ universes associated with it, each update to an instance at level $q$ causes at most $k \cdot (k \log n)^{\Oh(1)} = (k \log n)^{\Oh(1)}$ updates propagated to level $q-1$. Hence, there are at most $(k \log n)^{\Oh(L)}$ updates throughout all data structures per update to $\D^{\mathrm{main}}$.
\end{proof}


\begin{lemma} \label{lem:time1}
	The total time of updating all data structures for a single update to $\D^{\mathrm{main}}$, excluding the time of all reinitializations, can be bounded by $n^{\frac1L} \cdot 2^{(\frac{L}{\eps})^{\Oh(L)} + \Oh(L \log \log n)}$.
\end{lemma}

\begin{proof}
Let us focus first on the time required to recompute the solutions to the instances at level $1$ as explained in \cref{lem:csp-baker}, where $\eps', n$ and $\Delta$ from the statement of this lemma are equal to $\frac{\eps}{L}, n^{\frac{1}{L}}$ and $g(1) \in 2^{k^{\Oh(L)}}$, respectively.
Each such recomputation uses time $n^{\frac1L} \cdot \left(2^{k^{\Oh(L)}}\right)^{\Oh\left(\ceil{\frac{L}{\eps}}\right)} = n^{\frac1L} \cdot 2^{k^{\Oh(L)} \cdot \ceil{\frac{L}{\eps}}}$. As there are at most $(k \log n)^{\Oh(L)}$ updates by \cref{lem:updates-count}, the total time used for all such recomputations can be bounded by $(k \log n)^{\Oh(L)} \cdot n^{\frac1L} \cdot 2^{k^{\Oh(L)} \cdot \ceil{\frac{L}{\eps}}}$.

All overheads coming from processing a single update to some $\D$ like indexing the tables or navigating in $F$ are of the form $(k \log n)^{\Oh(1)}$. Hence, based on that and \cref{lem:updates-count}, excluding the time required for all hypothetical reinitializations, the total time needed for updating all the necessary information is $(k \log n)^{\Oh(L)} + (k \log n)^{\Oh(L)} \cdot n^{\frac1L} \cdot 2^{k^{\Oh(L)} \cdot \ceil{\frac{L}{\eps}}}$. Recall that we actually set $k = \ceil{\frac{L}{\eps}}$, so we can do the following simplifications: $(k \log n)^{\Oh(L)} = (\ceil{\frac{L}{\eps}} \log n)^{\Oh(L)} = 2^{\Oh(L(\log \log n + \log \frac{L}{\eps}))}$, hence $(k \log n)^{\Oh(L)} \cdot n^{\frac1L} \cdot 2^{k^{\Oh(L)} \cdot \ceil{\frac{L}{\eps}}} = 2^{\Oh(L(\log \log n + \log \frac{L}{\eps}))} \cdot n^{\frac1L} \cdot 2^{(\frac{L}{\eps})^{\Oh(L)}} = n^{\frac1L} \cdot 2^{(\frac{L}{\eps})^{\Oh(L)} + \Oh(L \log \log n)}$.
\end{proof}

As the next step, we are going to set epoch lengths and bound the amortized time of all reinitializations per single update to $\D^{\mathrm{main}}$.
Let $\tau_q$ denote the epoch length for data structures on level $q$ for some $2 \le q \le L$. Recall that the epoch length for a data structure is measured in the number of updates to this particular instance (as opposed to $\D^{\mathrm{main}}$). Also recall \cref{item:compressreq-size} from \cref{lem:indset-dynamic-compression}, which asserts that each update to $\D$ at level $q$ generates $(k \log n)^{\Oh(1)}$ updates to its children structures. Let us be more specific and let $c$ be such a constant that this number is at most $(k \log n)^c$. Then, we set $\tau_q = n^{\frac{q - 1}{L}} / (k \log n)^c$. For such a choice, it is indeed the case that the children data structures are passed at most $n^{\frac{q-1}{L}}$ updates before they are rebuilt, hence we maintain invariant \ref{item:indset-struct-size}: instances at level $q-1$ are of size at most $n^{\frac{q-1}{L}}$ at all times.


Thus, we can bound the amortized time complexity of initializations and reinitializations.

\begin{lemma} \label{lem:time2}
	The amortized time of all initializations and reinitializations per single update to $\D^{\mathrm{main}}$ is $n^{\frac1L} \cdot 2^{(\frac{L}{\eps})^{\Oh(L)} + \Oh(L \log \log n)}$.
\end{lemma}
\begin{proof}
We view that a lifespan of a particular data structure $\D$ corresponds to one epoch of the parent data structure (unless it is $\D^{\mathrm{main}}$). If the parent data structure of $\D$ gets rebuilt, we trash $\D$. Let $t$ be the number of started epochs of $\D$ throughout its whole lifetime. We note that when it is initialized, it is of a form $I \{\emptyset\}$ and consists of a single vertex (unless it is the initialization of $\D^{\mathrm{main}}$, which takes $n \cdot 2^{\Oh(k)}$ time). Hence, the first initialization takes constant time. As guaranteed by \cref{lem:indset-dynamic-compression} and \cref{lem:g-bound}, each reinitialization of $\D = (G, D, \rev, C)$ takes $(|V(G)| + |E(G)|) \log |V(G)| \cdot g(q)^{\Oh(k)} = (|V(G)| + |E(G)|) \log |V(G)| \cdot 2^{k^{\Oh(L)}}$ time. Thanks to our invariants, we are guaranteed that $|V(G)| \le n^{\frac{q}{L}}$ and $|E(G)| = \Oh(|V(G)|)$, hence the reinitialization time can be bounded as $n^{\frac{q}{L}} \log n \cdot 2^{k^{\Oh(L)}}$.

Note that if $t$ epochs were started, then the lifetime of $\D$ contained $t-1$ full epochs. The number of reinitializations will also be equal to $t-1$. Hence, if $u$ denotes the number of updates to $\D$ so far, then the time required for all reinitializations of $\D$ can be bounded as $(t-1) \cdot n^{\frac{q}{L}}\log n \cdot 2^{k^{\Oh(L)}} \leq (t-1) \cdot \tau_q \cdot n^{\frac{1}{L}} \cdot (k \log n)^{c+1} \cdot 2^{k^{\Oh(L)}} \le u \cdot n^{\frac{1}{L}} \cdot (\log n)^{\Oh(1)} \cdot 2^{k^{\Oh(L)}}$. Hence, the amortized time required for all reinitializations of $\D$ can be bounded as $n^{\frac{1}{L}} \cdot (\log n)^{\Oh(1)} \cdot 2^{k^{\Oh(L)}}$ per an update to $\D$. 
As each update to $\D^{\mathrm{main}}$ causes $(k \log n)^{\Oh(L)}$ updates to all structures in total and there are at most $k^L$ data structures, the total amortized time required for all reinitializations per a single update to $\D^{\mathrm{main}}$ is $n^{\frac{1}{L}} \cdot (\log n)^{\Oh(L)} \cdot 2^{k^{\Oh(L)}} \cdot k^{\Oh(L)} = n^{\frac1L} \cdot (\log n)^{\Oh(L)} \cdot 2^{(\frac{L}{\eps})^{\Oh(L)}} = n^{\frac1L} \cdot 2^{(\frac{L}{\eps})^{\Oh(L)} + \Oh(L \log \log n)}$.
\end{proof}

With these lemmas, we can present the final complexity analysis:
\begin{lemma}
	\label{lem:time3}
	 $L$ may be set so that the amortized time of an update to $\D^{\mathrm{main}}$ is $n^{\Oh\left( \frac{\log \log \log n}{\log \log n \cdot \eps} \right)} = n^{o\left( \frac{1}{\eps} \right)}$.
\end{lemma}
\begin{proof}
As both \cref{lem:time1,lem:time2} guarantee $n^{\frac1L} \cdot 2^{(\frac{L}{\eps})^{\Oh(L)} + \Oh(L \log \log n)}$ time for the corresponding computations, we get that the amortized time of an update to $\D^{\mathrm{main}}$ is equal to $n^{\frac1L} \cdot 2^{(\frac{L}{\eps})^{\Oh(L)} + \Oh(L \log \log n)}$ as well.

Let us set $L =  \frac{\log \log n}{\log \log \log n} \cdot \eps \cdot \delta$ for some absolute constant $\delta \in (0, 1)$ independent on $\eps$. For such a choice we have that $\Oh(L \log \log n) \leq \left( \frac{L}{\eps} \right)^{\Oh(L)}$, so the final complexity can be bounded by 
\begin{align*}
n^{\frac1L} \cdot 2^{\left(\frac{L}{\eps}\right)^{\Oh(L)}} & = n^{\frac{\log \log \log n}{\log \log n \cdot \eps \cdot \delta}} \cdot 2^{\left(\frac{\log \log n}{\log \log \log n}\cdot \delta\right)^{\Oh\left(\frac{\log \log n}{\log \log \log n} \cdot \eps \cdot \delta\right)}} \\ 
& = 2^{\frac{\log n \cdot \log \log \log n}{\log \log n \cdot \eps \cdot \delta}} \cdot 2^{2^{\Oh(\log \log n \cdot \eps \cdot \delta)}} \\
& = 2^{\frac{\log n \cdot \log \log \log n}{\log \log n \cdot \eps \cdot \delta}} \cdot 2^{(\log n)^{\Oh(\eps \cdot \delta)}}.
\end{align*}
If $\delta$ is sufficiently small, then $\Oh(\eps \cdot \delta)$ can be bounded from above by $\frac{1}{2}$. Then, the second factor will be dominated by the first one and the final complexity is then bounded by $n^{\Oh\left( \frac{\log \log \log n}{\log \log n \cdot \eps} \right)} = n^{o\left( \frac{1}{\eps} \right)}$.
\end{proof}

However, the time complexity of $f(\eps) \cdot n^{o(1)}$ would be preferable over $n^{o \left( \frac{1}{\eps} \right)}$, so as the last step, let us explain how to arrive at it. We will distinguish two cases, based on whether $n > 2^{2^{\frac{1}{\eps^2}}}$. 
\begin{enumerate}
	\item[Case 1:] $n > 2^{2^{\frac{1}{\eps^2}}}$.
	
	Then we have $\log \log n > \frac{1}{\eps^2} \Rightarrow \eps > \sqrt{\log \log n}$, so $n^{\Oh\left( \frac{\log \log \log n}{\log \log n \cdot \eps} \right)} = n^{\Oh\left( \frac{\log \log \log n}{\sqrt{\log \log n }} \right)} = n^{o(1)}$. 
	
	\item[Case 2:] $n \le 2^{2^{\frac{1}{\eps^2}}}$.
	
	In this case, instead of using our final algorithm, after every single update we can simply use a brute-force method provided by \cref{lem:csp-baker} and approximately solve the instance in time $n \cdot 2^{\Oh \left(\frac{1}{\eps} \right)} = 2^{2^{\Oh\left(\frac{1}{\eps^2}\right)}}$.
\end{enumerate} 
In any case, the time complexity of an~update can be bounded as $2^{2^{\Oh\left(\frac{1}{\eps^2}\right)}} \cdot n^{\Oh\left( \frac{\log \log \log n}{\sqrt{\log \log n }} \right)} = f(\eps) \cdot n^{o(1)}$.
Then, the initialization of the data structure on an~$n$-vertex graph $G \in \Cc$ can be performed in time $n^{1 + o(1)}$: we create an~instance of the data structure for an~edgeless $n$-vertex graph and add edges to it one by one.
Since $|E(G)| \leq \Oh(n)$ for any $G \in \Cc$, the bound on the initialization time follows.
This concludes the proof of \cref{thm:main} in the setting of {\sc Maximum Weight Independent Set}.



\section{Minimum Weight Dominating Set}

In this section, we describe a~dynamic approximation scheme for {\sc Minimum Weight Dominating Set}.
As in the case of the {\sc Maximum Weight Independent Set}, we assume that a~given apex-minor-free class of graphs $\Cc$ is fixed.
Moreover, throughout this section we denote by $\eps > 0$ the parameter $\eps$ fixed in the initialization of the data structure.
Consider the dynamic setting of \MinGeneralizedDom.
%
We assume that a~data structure for Generalized Domination is initialized with integers $s, d \geq 1$ and a~dynamic $(s,d)$-decent instance of \MinGeneralizedDom.
We consider the following types of updates to instances of \MinGeneralizedDom:

\begin{itemize}
  \item $\AddVertex(u, D_u, \cost_u)$: adds an~isolated vertex, say $u$, to the instance, together with a domain $D_u$ of size at most $d$ and a cost function $\cost_u \,\colon\, D_u \to \R_{\ge 0} \cup \{+\infty\}$.
  For each state $s \in D_u$, we initialize $\supply_u(s) = \emptyset$ and $\demand_u(s) = \emptyset$.
  
  \item $\AddEdge(u, v, D_u^{\supply}, D_u^{\demand}, D_v^{\supply}, D_v^{\demand})$: adds an~edge $e$ with endpoints $u$ and $v$; for each endpoint $w \in \{u, v\}$, the edge $e$ is added to each set $\supply_w(x)$ for $x \in D_w^{\supply}$ and to each set $\demand_w(x)$ for $x \in D_w^{\demand}$.
    
  \item $\RemoveEdge(e)$: removes an~edge $e$ from the graph, removing it also from the corresponding demand and supply sets.
  
  \item $\UpdateCost(u, \cost'_u)$: replaces $\cost_u$ for a~vertex $u \in V(G)$ with a~new function $\cost'_u \,\colon\, D_u \to \R_{\ge 0} \cup \{+\infty\}$.
\end{itemize}

As in the case with the independent set, we do not support vertex removals.


We now adapt the statement of \cref{thm:main} to the language of instances of \MinGeneralizedDom.

\begin{lemma}
\label{lem:domset-general-main}
Let $s, d \in \N$ be absolute constants and choose $\delta > 0$.
  Then there exists a~data structure storing a~dynamic $(s, d)$-decent instance $I = (G, D, \cost, \supply, \demand)$ of \MinGeneralizedDom under the assumption that $G \in \Cc$ holds at every point of time.
  The data structure maintains a~nonnegative real $p$ satisfying $(1 - \delta)\OPT \le p \le \OPT$, where $\OPT$ is the minimum possible cost of a~solution to $I$.
  The data structure is initialized in time $f(\delta) \cdot n^{1+o(1)}$, and each update takes time $f(\delta) \cdot n^{o(1)}$, where $f(\delta)$ is doubly-exponential in $\Oh(1/\delta^2)$.
\end{lemma}

Counterintuitively, instead of trying approximate \MinGeneralizedDom by finding an~approximate solution that would have a~slightly higher cost than optimal, the data structure of \cref{lem:domset-general-main} will maintain a~good \emph{lower bound} on the minimum possible cost of an~instance.
This is, however, still enough to show that the query $\QueryMWDS$ of \cref{thm:main} can be answered correctly: assume, for some absolute constant $\Delta \in \N$ that we are given a~dynamic graph $G$ with a~bound $\Delta$ on the maximum degree of a~vertex in $G$ and $\eps > 0$, and we wish to maintain a~nonnegative real $p'$ such that $\OPT \le p' \le (1 + \eps)\OPT$, where $\OPT$ is the minimum weight of a~dominating set of $G$.
  Set $\delta \coloneqq \frac{\eps}{1 + \eps}$ and instantiate the data structure of \cref{lem:domset-general-main}, initializing it with the $(\Delta, \Delta+1)$-decent instance $I$ of \MinGeneralizedDom constructed from $G$ and the parameter $\delta$.
  Each update to $G$ is translated to the appropriate update to $I$ and forwarded to the constructed data structure.
  Thus, the data structure is created in time $f(\delta) \cdot n^{1+o(1)}$ and each update to $G$ takes time $f(\delta) \cdot n^{o(1)}$.
  Since $1/\delta = 1 + 1/\eps$ and $f$ is doubly-exponential in $\Oh(1/\delta^2)$, $f$ is also doubly-exponential in $\Oh(1/\eps^2)$.

  The data structure from \cref{lem:domset-general-main} maintains a~nonnegative real $p$ such that $(1 - \delta)\OPT \le p \le \OPT$.
  Setting $p' \coloneqq \frac{p}{1 - \delta}$, we obtain that $\OPT \le p' \le (1-\delta)^{-1} \OPT = (1+\eps) \OPT$.
  Thus our implementation can return $p'$ as the over-approximation of the minimum weight of a~dominating set of $G$ that is at most a~factor of $\eps$ away from the optimum weight.

\subsection{Static variant of Generalized Domination}

Here, we show how to apply the Baker's technique to decent instances of \MinGeneralizedDom.

\begin{lemma}
  \label{lem:domset-bounded-tw}
  Let $s, d, w \ge 1$.
  Given on input an~$(s,d)$-decent instance of \MinGeneralizedDom $I = (G, D, \cost, \supply, \demand)$ with $n \ge 1$ vertices with the property that $\tw(G) \le w$, we can compute the minimum cost of a~solution to $I$ in worst-case time $\Oh(ns) \cdot d^{\Oh(w)}$.
\end{lemma}
\begin{proof}[Proof of \cref{lem:domset-bounded-tw}]
  Compute a~tree decomposition $T$ of $G$ of width $\Oh(w)$ in time $n \cdot 2^{\Oh(w)}$, e.g., using the algorithm of Korhonen \cite{Korhonen21}.
  Then the problem can be modeled by a~straightforward dynamic programming scheme on tree decompositions; here, we only give the states of the dynamic programming process.

  For every node $t \in V(T)$ of the tree decomposition with a~bag $\bag(t) \subseteq V(G)$, let $X_t \coloneqq \prod_{v \in \bag(t)} D_v$.
  For each $t \in V(T)$ and $x \in X_t$, we compute the value $P[t][x]$: the minimum cost of a~locally correct valuation $\phi$ of the vertices of $G$ in the subtree of $t$ in $T$ with the property that for every $v \in \bag(t)$, we have that $\phi(v) = x(v)$.
  
  For each node $t \in V(T)$, there exist at most $d^{|\bag(t)|} \leq d^{\Oh(w)}$ different dynamic programming states.
  It is then straightforward to implement the entire dynamic programming scheme in time $d^{\Oh(w)} \cdot \Oh(sn)$; here, the additional factor $\Oh(s)$ in the time complexity comes from the need of the verification of the demands and the supplies for each edge of the graph.
\end{proof}

\begin{lemma}
  \label{lem:domset-baker-partitioning}
  Let $s, d, k \ge 1$. Consider an~$(s,d)$-decent instance of \MinGeneralizedDom $I = (G, D, \cost, \supply, \demand)$ of optimum cost $\OPT$.
  Let $V_0, \dots, V_{k-1} \subseteq V(G)$ be so that the sets $N_G[V_0], N_G[V_1], \dots, N_G[V_{k-1}]$ are pairwise disjoint.
  For each $j \in \{0, \dots, k - 1\}$, consider the instance $I_j \coloneqq \Clear(I; V_j)$.
  Then:
  \begin{itemize}
    \item for every $j \in \{0, \dots, k-1\}$, $I_j$ is $(s,d)$-decent and has a~solution of cost at most $\OPT$;
    \item there exists $j \in \{0, \dots, k-1\}$ such that the minimum-cost solution to $I_j$ has cost at least $(1 - \frac{1}{k}) \OPT$. 
  \end{itemize}
\end{lemma}
\begin{proof}
  By the properties of $V_j$-cleared instances, we immediately have that each $I_j$ has a~solution of cost at most $\OPT$.
  It remains to show that for some $j \in \{0, \dots, k-1\}$, some subinstance $I_j$ has the minimum cost of a~solution lower-bounded by $(1 - \frac1k) \OPT$.

  Let $\phi$ be some optimum solution to $I$.
  By the pigeonhole principle and the fact that the closed neighborhoods of all sets $V_j$ are pairwise disjoint, there exists some $j \in \{0, \dots, k-1\}$ such that $\cost(\phi|_{N[V_j]}) \leq \frac1k \OPT$.
  We claim that $\OPT_j$, the minimum cost of a~solution of $I_j$, is at least $(1 - \frac{1}{k})\OPT$.
  
  Suppose we have a~solution $\psi_j$ to $I_j$ of cost strictly less than $(1 - \frac{1}{k})\OPT$.
  Recall that $V(I_j) = V(I)$ and that each $\psi_j(v)$ for each $v \in V(I)$ is also an~element of the domain of $v$ in $I$.
  With that in mind, construct a~valuation $\phi'$ of $V(I)$ as follows:
  \begin{itemize}
    \item if $v \in V_j$, set $\phi'(v)$ to the combination of the state $\phi(v)$ with $\psi_j(v)$;
    \item if $v \in N(V_j)$, set $\phi'(v)$ to the combination of the state $\psi_j(v)$ with $\phi(v)$;
    \item otherwise, set $\phi'(v) \coloneqq \psi_j(v)$.
  \end{itemize}
  
  We now show that the construction of $\phi'$ contradicts our assumptions:
  \begin{claim}
    \label{cl:domset-baker-contradiction}
    $\phi'$ is a~valid solution to $I$ of cost strictly less than $\OPT$.
  \end{claim}
  \begin{claimproof}
    Suppose $\phi'$ is not a valid solution to $I$.
    Then there exists an~edge $e$ with endpoints $u$, $v$ such that $e \in \demand_v(\phi'(v))$, but $e \notin \supply_u(\phi'(u))$.
    By the definition, the edge cannot have one endpoint in $V_j$ and the other outside of $N[V_j]$.    
    We now consider cases, depending on the location of $u$ and $v$ with respect to $V_j$:
    \begin{itemize}
      \item Assume $v \in V_j$. By the state-monotonicity of $v$, we have $\demand_v(\phi'(v)) \subseteq \demand_v(\phi(v))$, so $e \in \demand_v(\phi(v))$.
      Then also $u \in N[V_j]$. By the state-monotonicity of $u$, we have that $\supply_u(\phi'(u)) \supseteq \supply_u(\phi(u))$, so $e \notin \supply_u(\phi(u))$.
      This is, however, a~contradiction: $e$ witnesses that $\phi$ was not a~valid solution to $I$ in the first place.
      
      \item Assume $v \notin V_j$. Then $e \in \demand_v(\psi_j(v))$: either $v \in N(V_j)$ and then this follows from the state-monotonicity of $v$ (so $\demand_v(\phi'(v)) \subseteq \demand_v(\psi_j(v))$), or $v \notin N[V_j]$ and then simply $\phi'(v) = \psi_j(v)$.
      Also, from the state-monotonicity of $u \in V(G)$ we have that $\supply_u(\phi'(u)) \supseteq \supply_u(\psi_j(u))$, so $e \notin \demand_u(\psi_j(u))$.
      Since $v \notin V_j$, we have $e \in E(I_j)$, so this yields a~contradiction: $e$ witnesses that $\psi_j$ was not a~valid solution to $I_j$ in the first place.
    \end{itemize}
  Hence $\phi'$ is a~correct solution to $I$.
  By the state-monotonicity of vertices of $I$, and by the facts that $\cost(\phi|_{N[V_j]}) \leq \frac1k \OPT$ and $\psi_j$ has cost less than $(1 - \frac{1}{k})\OPT$, we conclude that $\phi'$ has cost strictly smaller than $\OPT$.
  \end{claimproof}

  \Cref{cl:domset-baker-contradiction} yields a~contradiction.
  Thus $I_j$ has no solution of cost smaller than $(1 - \frac1k)\OPT$.
\end{proof}

\begin{lemma}
  \label{lem:domset-static-baker}
  Let $s, d \ge 1$, $\delta > 0$, and consider an~$(s,d)$-decent instance of $I \in \Cc$ of \MinGeneralizedDom with $n \ge 1$ vertices and optimum solution cost $\OPT$.
  Then we can compute a~nonnegative real $p$ such that $(1 - \delta)\OPT \le p \le \OPT$ in time $\Oh(ns) \cdot d^{\Oh(1/\delta)}$.
\end{lemma}
\begin{proof}
  Assume without loss of generality that $\delta \in (0, 1)$, and set $k \coloneqq \left\lceil \frac{1}{\delta} \right\rceil$.
  Assume that the instance $I$ is connected; otherwise, solve for each connected component separately and return the sum of the results.
  Fix a~vertex $r \in V(I)$ and create in time $\Oh(n)$ the partitioning $V(I) = A_0 \cup \dots \cup A_{4k-1}$, where $A_i$ contains vertices $v$ such that $\mathsf{dist}(r, v) \equiv i \mod{4k}$.
  Then, for each $j \in \{0, \dots, k-1\}$, let $V_j = A_{4j+1} \cup A_{4j+2}$ and $I_j = \Clear(I; V_j)$.
  Note that $N[V_j] \subseteq A_{4j} \cup A_{4j+1} \cup A_{4j+2} \cup A_{4j+3}$, so the closed neighborhoods of the sets $V_j$ are pairwise vertex-disjoint and \cref{lem:domset-baker-partitioning} applies.
  That is, if $\OPT$ is the minimum-cost solution to $I$, then each $I_j$ is $(s,d)$-decent and has a~solution of cost at most $\OPT$, and there exists a~subinstance with minimum cost of a~solution at least $(1 - \frac1k) \OPT \geq (1 - \delta)\OPT$.
  
  Observe that for each $j \in \{0, \dots, k-1\}$, the Gaifman graph of $I_j$ has bounded treewidth: consider a~connected component $C$ of the Gaifman graph.
  The component cannot simultaneously contain vertices of $A_{4j+1}$ and $A_{4j+2}$, so the vertex set of $C$ is contained within some $4k$ consecutive BFS layers.
  In other words, for some $\ell \geq 0$, we have that $C \subseteq \{v \in V(G) \,\colon\, \ell \leq \mathsf{dist}(s, v) \leq \ell + 4k - 1\}$.
  Thus by \cref{lem:csp-partitioning}, $G[C]$ has treewidth at most $\Oh(k)$, so also the Gaifman graph of $I_j$ has treewidth at most $\Oh(k) = \Oh(\frac{1}{\delta})$.

  Let $\OPT$ be the minimum-cost solution to $I$.
  By \cref{lem:domset-baker-partitioning}, each $I_j$ is $(s,d)$-decent, the treewidth of the Gaifman graph of $I_j$ is at most $\Oh(\frac{1}{\delta})$, all instances have solutions of cost at most $\OPT$, and there exists an~instance with the minimum cost of a~solution at least $(1 - \frac{1}{k}\OPT) \geq (1 - \delta)\OPT$.
  Now, we solve each instance $I_j$ optimally using the algorithm from \cref{lem:domset-bounded-tw}.
  Let $p_j$ be the value returned by $I_j$ and let $p = \max(p_0, \dots, p_{k-1})$; then $p \in [(1 - \delta)\OPT, \OPT]$, so $p$ is as~required.
  It can be easily verified that the algorithm runs in time $n \cdot s \cdot d^{\Oh(1 / \delta)}$.
\end{proof}

\subsection{Compression for Generalized Domination}

We proceed to show how to translate the process of instance compression to the setting of \MinGeneralizedDom.

We begin by presenting how to encode interactions between a~set of vertices and its neighborhood in a~concise way. Consider an~instance $I = (G, D, \cost, \supply, \demand)$ and a~set of vertices $R \subseteq V(I)$.
Let also $S$ be the neighborhood of $R$, i.e., $S = N_G(R)$.
For every valuation $\phi \in \prod_{u \in R} D_u$ of $R$, we define the \emph{interaction} of $R$ with $S$:
\[
  \mathsf{interaction}_{R,S}(\phi) \in \left(2^{\{\Supply, \Demand\}}\right)^{E(R, S)},
\]
so that
$\interaction_{R,S}(\phi)$ is a~function $x$ mapping each edge connecting $R$ with $S$ to a~subset of $\{\Supply, \Demand\}$ as follows: let $e \in \delta(S)$, $e = uv$, where $u \in R$ and $v \in S$.
Then we have $\Supply \in x(e)$ if and only if $e \in \supply_u(\phi_u))$, and $\Demand \in x(e)$ if and only if $e \in \demand_u(\phi(u))$.

In other words, we record, for every edge $e$ connecting a~vertex $u \in R$ with a~neighbor $v \in S$, whether $u$ provides supply on $e$ (so that in any solution extending $\phi$ to $I$, $v$ can safely demand that supply) and whether $u$ demands supply on $e$ (so that in a~solution extending $\phi$ to $I$, $v$ has to provide that supply).
Intuitively, $\interaction_{R,S}(\phi)$ stores the entire interaction of $R$ with the neighborhood in a~concise way.
The following lemma formalizes this intuition:
\begin{lemma}
  \label{lem:domset-interaction-substitution}
  Let $R \subseteq V(I)$ and $S = N(R)$.
  Assume $\phi$ is a~correct solution to $I$ and $\psi$ is a~locally correct solution on $R$ so that $\interaction_{R, S}(\phi|_R) = \interaction_{R, S}(\psi)$.
  Then, the following valuation $\phi'$ of $V(I)$:
  \[ \phi'(u) = \begin{cases}
     \psi(u) & \text{if }u \in R, \\
     \phi(u) & \text{if }u \notin R
  \end{cases} \]
  is also a~correct solution to $I$.
\end{lemma}
\begin{proof}
  Assume otherwise; in this case, there exists an~edge $e \in E(G)$ between $u$ and $v$ so that $e \in \demand_v(\phi'(v))$, yet $e \notin \supply_u(\phi'(u))$.
  It cannot be that $u, v \in R$ since $\phi'|_R = \psi$ is locally correct on $R$; and also it cannot be that $u, v \notin R$ as $\phi'|_{V(I) \setminus R} = \phi|_{V(I) \setminus R}$ and $\phi$ is a~correct solution to $I$.
  Hence, exactly one of the endpoints of $e$ is in $R$.
  
  If $u \in R$, then $v \in S$.
  Since $\phi'(u) = \psi(u)$, we have that $e \notin \supply_u(\psi(u))$ and therefore  $\Supply \notin \interaction_{R,S}(\psi)(e)$.
  So also $\Supply \notin \interaction_{R,S}(\phi|_R)(e)$, implying $e \notin \supply_u(\phi(u))$.
  But $\phi'(v) = \phi(v)$, so $e \in \demand_v(\phi'(v))$ implies $e \in \demand_v(\phi(v))$.
  This is a~contradiction since $e$ witnesses that $\phi$ was not a~correct solution to $I$ in the first place.

  The other case is that $v \in R$ and $u \in S$. From $\phi'(v) = \psi(v)$ we have $e \in \demand_v(\psi(v))$, so $\Demand \in \interaction_{R,S}(\psi)(e)$, so $\Demand \in \interaction_{R,S}(\phi|_R)(e)$ and $e \in \demand_v(\phi(v))$. Also, since $\phi'(u) = \phi(u)$, then we have $e \notin \supply_u(\phi(u))$ -- a~contradiction as $e$ witnesses that $\phi$ was not a~correct solution to $I$.
%
\end{proof}

We are now ready to give the definition of a~compressed instance.


\begin{definition}
  Let $I = (G, D, \cost, \supply, \demand)$ be an~$(s,d)$-decent instance of \MinGeneralizedDom and $Y \subseteq V(I)$.
  Then, we say that a~\emph{compressed instance} is an~instance $I\{Y\} = (G', D', \cost', \supply', \demand')$ created from $I$ by the following series of steps:
  
  \begin{enumerate}
    \item Enumerate all connected components $C_1, \dots, C_\ell \in \cc(G \setminus Y)$.
    Assume that, for some $\ell_1 \in \{0, \dots, \ell\}$, the neighborhood of each component $C_1, \dots, C_{\ell_1}$ intersects $Y$, and the neighborhood of each component $C_{\ell_1 + 1}, \dots, C_{\ell}$ is disjoint from $Y$ (that is, each $C_{\ell_1 + 1}, \dots, C_{\ell}$ is a~connected component of $G$ not adjacent to $Y$).
    
    \item For each $i \in \{1, \dots, \ell_1\}$, let $S_i = N(C_i) \subseteq Y$. Collapse $C_i$ to a~single vertex $u_i$ with the domain
    \[ D_{u_i} \coloneqq \left(2^{\{\Supply, \Demand\}}\right)^{E(C_i, S_i)}. \]
    
    For each $x \in D_{u_i}$, we set $\cost_{u_i}(x)$ to the minimum cost of a~valuation $\phi \in \prod_{u \in C_i} D_u$, locally correct on $C_i$, so that $x = \interaction_{C_i, S_i}(\phi)$.
    If no such valuation exists, we set $\cost_{u_i}(x) = +\infty$.
    
    Finally, set the supply and the demand of any given state $x$ as follows:
    \begin{align*}
    \supply'_{u_i}(x) &= \{e \in E(C_i, S_i) \,\colon\, \Supply \in x(e)\}, \\
    \demand'_{u_i}(x) &= \{e \in E(C_i, S_i) \,\colon\, \Demand \in x(e)\}.
  \end{align*}
  
  \item Collapse $R \coloneqq C_{\ell_1 + 1} \cup \dots \cup C_\ell$ to a~single isolated vertex $u_{\ominus}$ with the one-state domain $D_{u_{\ominus}} = \{\ominus\}$; the cost $\cost_{u_{\ominus}}(\ominus)$ is defined as the minimum-cost valuation of $R$ that is locally correct on $R$ (note that a~valuation locally correct on $R$ with finite cost always exists).
    \end{enumerate}
\end{definition}

Intuitively, each component $C_i \in \cc(G \setminus Y)$ with nonempty neighborhood $S$ on $Y$ is collapsed into a single vertex $u_i$; note that, in contrast to the compression for Independent Set, different components with the same neighborhood $S$ are collapsed to separate vertices.
The collapsed vertex encodes in its state all possible interactions of $C_i$ with $S_i$.
Note that by \cref{lem:domset-interaction-substitution}, it is enough to record, for each possible interaction, the minimum-cost valuation of $C_i$ with this interaction on $S_i$.
On the other hand, the components not adjacent to $Y$ do not interact with $Y$ in any way, so the union $R$ of these components can be replaced with a~single vertex representing the minimum-cost locally correct valuation of $R$.

Observe that the construction above is correct in the sense that in $I\{Y\}$, every vertex $u$ has a~state $s_u$ with finite cost and full supply $\supply_u(s_u) = \delta(u)$: for $u \in Y \cup \{\ominus\}$ this is obvious.
Then, for a~vertex $u_i$ with $i \in \{1, \dots, \ell_1\}$, there exists a~finite-cost valuation $\phi_i$, locally correct on $C_i$, mapping each vertex $u \in C_i$ to $s_u$.
Moreover, for this valuation, we have $\Supply \in \interaction_{C_i, S_i}(\phi_i)(e)$ for all $e \in E(C_i, S_i)$.

We now prove a~series of properties of compressed instances.
We begin by proving that compression does not change the minimum cost of an~instance:

\begin{lemma}
  For every instance $I$ of \MinGeneralizedDom and $Y \subseteq V(I)$, the instances $I$ and $I\{Y\}$ have the same minimum cost of a~solution.
\end{lemma}
\begin{proof}
  Let $\OPT$ be a~minimum cost of a~solution to $I = (G, D, \cost, \supply, \demand)$ and $\OPT'$ be a~minimum cost of a~solution to $I\{Y\} = (G', D', \cost', \supply', \demand')$.
  Also let $\alpha \,\colon\,V(I) \to V(I\{Y\})$ be a~function mapping vertices $v \in V(I)$ onto elements of $V(I\{Y\})$ to which $v$ was collapsed during the compression.
  Note that $\alpha(v) = v$ for $v \in Y$.
  
  Let $\phi$ be any~minimum-cost solution to $I$.
  We construct a~correct solution $\phi'$ to $I\{Y\}$ as follows.
  \begin{itemize}
    \item If $u \in Y$, then set $\phi'(u) \coloneqq \phi(u)$.
    \item If $u \notin Y$ and $u \neq u_{\ominus}$, then let $C = \alpha^{-1}(u) \in \cc(G \setminus Y)$, $S = N_G(C) \subseteq Y$, and set $\phi'(u) \coloneqq \interaction_{C, S}(\phi|_C)$.
    \item Set $\phi'(u_{\ominus}) = \ominus$.
  \end{itemize}
  
  The correctness of $\phi'$ follows easily from the definition.
  Note that for every $u \in Y$, we naturally have $\cost'_u(\psi(u)) = \cost_u(\phi(u))$.
  For $u \notin Y$ with $u \neq u_{\ominus}$, let $C = \alpha^{-1}(u) \in \cc(G \setminus Y)$ and $S = N_G(C)$.
  Since $\phi|_C$ is a~locally correct valuation of $C$ with interaction $\phi'(u)$ and $\cost'_u(\phi'(u))$ is the minimum cost of such a~valuation, we have $\cost'_u(\psi(u)) \leq \cost_u(\phi|_C)$.
  By the same token, $\phi|_{\alpha^{-1}(u_{\ominus})}$ is a~locally correct valuation of $\alpha^{-1}(u_{\ominus})$ and $\cost'_{u_{\ominus}}(\ominus)$ is the minimum cost of such a~valuation; thus, $\cost'_{u_{\ominus}}(\ominus) \leq \cost_u(\phi|_{\alpha^{-1}(u_{\ominus})})$.
  Hence the cost of $\phi'$ in $I\{Y\}$ is at most equal to the cost of $\phi$ in $I$.
  Thus $\OPT' \le \OPT$.
  
  \smallskip
  Now let $\psi'$ be any minimum-cost solution to $I\{Y\}$.
  Recover a~valid solution $\psi$ to $I$ as follows.
  \begin{itemize}
    \item If $u \in Y$, then set $\psi(x) \coloneqq \psi'(x)$.
    \item If $u \notin Y$ and $u \neq u_{\ominus}$, then let $C = \alpha^{-1}(u) \in \cc(G \setminus Y)$ and $S = N(C) \subseteq Y$.
    Let $\zeta_C \in \prod_{v \in C} D_v$ be the minimum-cost locally correct valuation of $C$ with $\interaction_{C, S}(\zeta_C) = \psi'(u)$, and set $\psi(x) \coloneqq \zeta_C(x)$ for $x \in C$. 
    \item Similarly, if $u = u_{\ominus}$, then let $R = \alpha^{-1}(u_{\ominus})$, and let $\zeta_R \in \prod_{v \in R}$ be the minimum-cost locally correct valuation of $R$; set $\psi(x) \coloneqq \zeta_R(x)$ for $x \in R$.
  \end{itemize}
  It can be verified through a~straightforward case study that $\phi$ is a~valid solution to $I$, and $\phi$ has the same total cost as $\psi$ in $I\{Y\}$.
  Therefore $\OPT \le \OPT'$.
\end{proof}

We proceed to show that, under suitable conditions, compressed instances are decent and belong to $\Cc$ if the original instance belonged to $\Cc$:

\begin{lemma}
  If $I \in \Cc$ and $Y \subseteq V(I)$, then $I\{Y\} \in \Cc$.
\end{lemma}
\begin{proof}
  As in \cref{lem:compression-stays-in-class}, we can show that the Gaifman graph of $I\{Y\}$ is a~minor of the Gaifman graph of $I$.
  Thus, $I\{Y\} \in \Cc$.
\end{proof}

\begin{lemma}
  Let $s, d, t \ge 1$.
  Consider an~$(s,d)$-decent instance $I = (G, D, \cost, \supply, \demand)$ and $Y \subseteq V(I)$.
  If for every connected component $C \in \cc(G \setminus Y)$, we have $|N(C)| \le t$, then
  $I\{Y\} = (G', D', \cost', \supply', \demand')	$ is $(st, d + 4^{st})$-decent.
\end{lemma}
\begin{proof}
  First observe that the set of edges incident to vertices of $Y$ does not change under the compression; thus $(s,d)$-meager vertices of $Y$ remain $(s,d)$-meager.
  For the same reason, each state-monotonous vertex of $Y$ remains such.
  
  Now consider a~vertex $u \in I\{Y\} \setminus (Y \cup \{u_{\ominus}\})$ created as a~result of a~collapse of a~component $C \in \cc(G \setminus Y)$ with neighborhood $S \subseteq Y$.
  By the definition of $u$, we see that $|D_u| = 4^{|E(C, S)|}$.
  Observe that $|E(C, S)| \leq |S| \cdot s$ as every vertex of $S$ is $(s,d)$-meager and thus is incident to at most $s$ edges.
  Also, $|S| \leq t$ from the assumption.
  Hence, $|D_u| \leq 4^{st}$.
  Also, it is easy to verify that the degree of $u$ in $I\{Y\}$ is equal to $|E(C, S)| \leq st$.
  
  It remains to verify that $u$ is state-monotonous.
  Recall that $D'_u = \left(2^{\{\Supply, \Demand\}}\right)^{|E(C, S)|}$.
  Consider two states $x_1, x_2 \in D'_u$, aiming to prove that there exists a~combination $x$ of $x_1$ with $x_2$.
  If $\cost'_u(x_1) + \cost'_u(x_2) = +\infty$, then it is enough to put as $x$ the constant function always returning $\{\Supply\}$; i.e., the state $x$ for which $\supply'_u(x) = \delta_{G'}(u)$ and $\demand'_u(x) = \emptyset$.
  Otherwise, there exist locally correct valuations $\phi_1$, $\phi_2$ of $C$, of cost exactly $\cost'_u(x_1)$ and $\cost'_u(x_2)$, respectively, so that $\interaction_{C, S}(\phi_1) = x_1$ and $\interaction_{C, S}(\phi_2) = x_2$.
  Construct a~new valuation $\phi$ of $C$ as follows: for each $v \in C$, let $\phi(v)$ be the combination of $\phi(x_1)$ with $\phi(x_2)$, and let $x \coloneqq \interaction_{C, S}(\phi)$.
  The following claims are straightforward and follow from the verification with the definitions.
  \begin{claim}
    $\phi$ is locally correct on $C$ and has cost at most $\cost'_u(x_1) + \cost'_u(x_2)$.
  \end{claim}
  \begin{claim}
    Let $e \in E(C, S)$. Then $\Supply \in x(e)$ if and only if $\Supply \in x_1(e)$ or $\Supply \in x_2(e)$. Moreover, if $\Demand \in x(e)$, then $\Demand \in x_1(e)$.
  \end{claim}
  It follows that $x$ is a~combination of $x_1$ with $x_2$.
  Since $x_1$, $x_2$ were chosen arbitrarily, we conclude that $u$ is state-monotonous.
  
  Finally, the vertex $u_{\ominus}$ is isolated and contains only one state in its domain; hence it is state-monotonous and $(1,1)$-decent.
\end{proof}

\subsection{Dynamic maintenance of compressions}

We now introduce the analog of \cref{lem:indset-dynamic-compression} in the setting of Generalized Domination.
Note that the proof of the following lemma follows closely that of \cref{lem:indset-dynamic-compression}, however we provide it here in full for completeness.

\begin{lemma}
  \label{lem:domset-dynamic-compression}
  Let $w, n, s, d \in \N$.
  One can construct a~data structure that supports the following operations:
  \begin{itemize}
    \item \textbf{Initialize} the data structure with an~$n$-vertex $(s,d)$-decent instance of \MinGeneralizedDom $I = (G, D, \cost, \supply, \demand)$, where $G \in \Cc$ and $\tw(G) \leq w$.
    \item \textbf{Update} the instance $I$ using one of the following update types: $\AddVertex$, $\AddEdge$, $\RemoveEdge$, $\UpdateCost$.
    It is guaranteed that after the update, we have that $G \in \Cc$ and the instance $I$ after the update is $(s,d)$-decent.
  \end{itemize}
  The initialization of the data structure is performed in time $2^{\Oh(sw)} \cdot d \cdot n \log n$.
  Afterwards, the data structure additionally maintains an~instance $I^\star = (G^\star, D^\star, \cost^\star, \supply^\star, \demand^\star)$ of \MinGeneralizedDom with the following properties:
  \begin{enumerate}[label=(\alph*)]
    \item \label{item:dom-compress-in-class} $G^\star \in \Cc$;
    \item \label{item:dom-compress-revenue} the minimum cost of a~solution to $I^\star$ is equal to that of $I$;
    \item \label{item:dom-compress-domain} $I^\star$ is $(\Oh(sw), d + 4^{\Oh(sw)})$-decent;
    \item \label{item:dom-compress-size} after a~sequence of $t \geq 0$ updates to $I$, we have $|V(I^\star)| \leq t \cdot s \cdot w^{\Oh(1)} \log n$.
    Moreover, on each update to $I$, the instance $I^\star$ can be updated in time $2^{\Oh(sw)} \cdot \log n$ and causes at most $s \cdot w^{\Oh(1)} \log n$ updates to $I^\star$.
  \end{enumerate}
\end{lemma}
\begin{proof}[Proof of \cref{lem:domset-dynamic-compression}]
  As in the case of the independent set, preprocess $G$ in time $2^{\Oh(w)} \cdot n \log n$ and produce an~elimination forest $F$ of $G$ with $V(F) = V(G)$ with the following properties:
  \begin{itemize}
    \item $F$ has height $\Oh(w \log n)$;
    \item for each $u \in V(F)$, we have $|\Reach_F(u)| = \Oh(w)$;
    \item for each $u \in V(F)$, the graph $G[\desc_F[u]]$ is connected.
  \end{itemize}  
  For convenience, for each $u \in V(G)$, let $C_u = \desc_F[u]$, $S_u = \Reach_F(u) = N(C_u)$ and $X_u = \left(2^{\{\Supply, \Demand\}}\right)^{E(C_u, S_u)}$.
  Observe that each element of $\interaction_{C_u, S_u}$ is an~element of $X_u$.
  Note that $|S_u| \le \Oh(w)$ and $I$ is $(s,d)$-decent, so $|E(C_u, S_u)| \leq \Oh(sw)$ and therefore $|X_u| \le 2^{\Oh(sw)}$.
  
  Observe also that the nodes in the elimination forest have bounded degrees:
  \begin{claim}
    \label{cl:domset-tree-bd-deg}
    Every node $u \in V(F)$ has at most $\Oh(sw)$ children in $F$.
  \end{claim}
  \begin{claimproof}
    Enumerate the children of $u$ in $F$: $c_1, \dots, c_t$.
    For each $c_i$, we have $\Reach_F(c_i) \subseteq \Reach_F(u) \cup \{u\}$.
    Each such $c_i$ is thus a~neighbor of some vertex in $\Reach_F(u) \cup \{u\}$.
    The claim follows since each vertex of $G$ has degree at most $s$ and that $|\Reach_F(u) \cup \{u\}| \leq \Oh(w)$.
  \end{claimproof}
  
  On initialization of the data structure, compute the following dynamic programming table:
  
  \begin{itemize}
    \item For each $u \in V(F)$ and every interaction $x \in X_u$, compute $T[u][x]$: the minimum possible cost of a~locally correct partial solution $\phi \in \prod_{x \in C_u} D_x$ such that $\interaction_{C_u, S_u}(\phi) = x$; or $+\infty$ if no such partial solution exists.
  \end{itemize}
  
  The table $T[\cdot][\cdot]$ has at most $n \cdot 2^{\Oh(sw)}$ states and can be computed in time $nd \cdot 2^{\Oh(sw)}$.
  Hence, the total preprocessing time is upper-bounded by $2^{\Oh(sw)} \cdot nd \log n$.
  Also observe that the table $T[\cdot][\cdot]$ suffices to compute the minimum-cost solution for $I$ at the time of initialization: the cost is equal to the sum over $T[r][\emptyset]$, ranging over all roots $r$ of trees of $F$.
  
  As before, in the beginning we set $A = B = Z = \emptyset$.\\
  Let $\Iinit = (\Ginit, D^{\mathrm{init}}, \cost^{\mathrm{init}}, \supply^{\mathrm{init}}, \demand^{\mathrm{init}})$ be the initial instance of \CSPs. In the sequel, by $\Icur = (\Gcur, D^{\mathrm{cur}}, \cost^{\mathrm{cur}}, \supply^{\mathrm{cur}}, \demand^{\mathrm{cur}})$ we denote the current snapshot of the instance  in the data structure; initially, $\Icur \coloneqq \Iinit$.
  
  Recall from the setting of the Independent Set that throughout the life of the data structure, we maintain the following invariants:
  \begin{itemize}
    \item $A = V(\Icur) \setminus V(\Iinit)$ is the set of vertices added to $I$ since the instantiation of the data structure.
    \item $B \subseteq V(\Icur)$ is the set of vertices that were part of any update to $I$ so far (i.e., $v \in B$ if $v$ was added to $I$, the cost of $v$ was changed, or an~edge incident to $v$ was added or removed).
    \item $Z = A \cup \anc_F[B \setminus A]$. In other words, $Z$ contains all vertices of $A$ and the ancestors in $F$ of all vertices of $B \setminus A$.
    \item $I^\star$ is equivalent to $\Icur\{Z\}$.
  \end{itemize}
  
  Now, the costs of the vertices in the compressed instance can be inferred from the entries of $T[\cdot][\cdot]$:
  \begin{itemize}
    \item The connected components of $\Gcur \setminus Z$ adjacent to $Z$ are exactly the components of the form $\desc_F[y]$, where $y \notin Z$ is not a~root of any tree in $F$ and $\parent_F(y) \in Z$.
    In the compression, for each $y$, we introduce one vertex $u_y \in I\{Z\} \setminus Z$ with the domain $D_{u_y} = X_y$ and the cost function $\cost_{u_y} \,\colon\, D_{u_y} \to \R_{\ge 0} \cup \{+\infty\}$, where $\cost_{u_y}(x) = T[y][x]$ for $x \in D_{u_y}$.
    \item The connected components of $\Gcur \setminus Z$ non-adjacent to $Z$ are exactly the components of the form $\desc_F[r]$, where $r \notin Z$ is a~root of a~tree in $F$.
    All these connected components are collapsed to a~single vertex $u_{\eps}$ with a~single state $\eps$, of cost that is the sum over $T[r][\emptyset]$ for all roots $r \notin Z$ of trees in $F$.
  \end{itemize}

  Recall the following claim from the setting of the Independent Set.
  Note that since the sets $A$, $B$, $Z$ are constructed as in that setting, the claim also holds here.  
  \begin{claim}
    \label{cl:domset-z-growth}
    On each update to $I$, the set $Z$ grows by at most $\Oh(w \log n)$ vertices.
  \end{claim}
  
  Two instances $I_1$, $I_2$ of \MinGeneralizedDom are \emph{equivalent} if the instances are isomorphic after the removal of all isolated vertices with the constant-zero cost function.
  We now sketch a~single update to $Z$ -- an~addition of an~appendix of $Z$.
  
  \begin{claim} \label{cl:domset-add-vtx-compr}
	Let two sets $Z_1, Z_2 \subseteq V(\Gcur)$ be such that $A \subseteq Z_1 \subseteq Z_2 \subseteq V(\Gcur)$, $|Z_2| = |Z_1| + 1$ and $Z_1 \setminus A$ and $Z_2 \setminus A$ are prefixes of $F$ with $Z_2 - Z_1 \subseteq V(F)$. Then, an instance equivalent to $\Icur\{Z_2\}$ can be obtained from an instance equivalent to $\Icur\{Z_1\}$ through a sequence of $sw^{\Oh(1)}$ updates.
	Moreover, this sequence can be computed in time $2^{\Oh(sw)}$.
  \end{claim}
  \begin{claimproof}
    Let $z \in V(F)$ be such that $Z_2 = Z_1 \cup \{z\}$.
	Let $\Cfrak_1$ be the set of connected components of $\Gcur \setminus Z_1$, $\Cfrak_2$ be the set of connected components of $\Gcur \setminus Z_2$ and $c_1, \ldots, c_t$ be the set of children of $z$ in $F$. We have that $\Cfrak_1 \setminus \Cfrak_2 = \{\Gcur[\desc_F[z]]\}$ and $\Cfrak_2 \setminus \Cfrak_1 = \{\Gcur[\desc_F[c_1]], \ldots, \Gcur[\desc_F[c_t]]\}$.
	
	Let $S = \Reach_F(z)$. The compressed instance $\Icur\{Z_1\}$ contains a vertex $v_S$ representing the union of all connected components of $\Gcur \setminus Z_1$ whose neighborhoods are exactly $S$; and $\Gcur[\desc_F[z]]$ is one of such components. To obtain $\Icur\{Z_2\}$ from $\Icur\{Z_1\}$, we need to:
	\begin{enumerate}
	  \item Remove the contribution of $\Gcur[\desc_F[z]] \in \Cfrak_1 \setminus \Cfrak_2$ from the compressed instance.
	  If $z$ is non-root, then the vertex $u_z$ corresponding to the collapse of $\desc_F[z]$ should be removed from the instance; this is emulated by removing all edges incident to $u_z$ and replacing the cost of $u_z$ with the zero function; this can be done using $\Oh(w)$ updates to the compressed instance.
	  Otherwise, $\desc_F[z]$ is collapsed to the special vertex $u_{\eps}$ (possibly with other connected components non-adjacent to $Z$).
      Then, $u_{\eps}$ has only one state $\eps$, and its cost is adjusted by simply subtracting $T[z][\emptyset]$.
	  \item Add the vertex $z$ to the compressed instance.
	  Since the set of neighbors of $z$ in $Z_1$ is exactly $\Reach_F(z)$, this requires one vertex addition and $\Oh(w)$ edge additions.
	  \item Include the contribution of the connected components $\Gcur[\desc_F[c_1]], \dots, \Gcur[\desc_F[c_t]] \in \Cfrak_2 \setminus \Cfrak_1$ in the compressed instance.
	  Since $t \leq \Oh(sw)$ by \cref{cl:domset-tree-bd-deg}, we simply iterate all components and for each, we add the corresponding vertex $u_{c_i}$ to the instance (with the costs of the states sourced from $T$), along with at most $\Oh(w)$ edges incident to each vertex.
	  Thus, the number of updates to the compressed instance is $\Oh(sw^2)$.
	  For each fresh vertex, we iterate over all $|X_{c_i}| = 2^{\Oh(sw)}$ states $u_{c_i}$ can attain, so the total time of the update is $\Oh(sw^2) \cdot 2^{\Oh(sw)} = 2^{\Oh(sw)}$. \hfill\qedhere
	\end{enumerate}
  \end{claimproof}
  
  As before, the implementation of the data structure directly follows from \cref{cl:domset-z-growth,cl:domset-add-vtx-compr}.
  Since on each update to $\Icur$, the set $Z$ is updated $\Oh(w \log n)$ times, the total number of updates performed on $I^\star$ is $\Oh(w \log n) \cdot sw^{\Oh(1)} = sw^{\Oh(1)} \log n$. Therefore, the total update time is $2^{\Oh(sw)} \log n$.
\end{proof}


\subsection{Full algorithm}

We are now ready to give the exposition of the data structure proving \cref{lem:domset-general-main}.
The rest of this section describes the data structure $\D^{\mathrm{main}}$ maintaining a~$(s, d)$-decent dynamic instance of \MinGeneralizedDom $I^{\mathrm{main}} = (G^{\mathrm{main}}, D^{\mathrm{main}}, \cost^{\mathrm{main}}, \supply^{\mathrm{main}}, \demand^{\mathrm{main}})$.
Similarly to the case of \CSP, we fix an~integer $L \in \N$ and set $k \coloneqq \left\lceil L/\delta \right\rceil$.
We construct a~recursive, $L$-level data structure comprised of auxiliary data structures maintaining subinstances of \MinGeneralizedDom; the only data structure at level $L$ maintains $I^{\mathrm{main}}$, and each data structure $\D$ at level $q \in \{2, 3, \dots, L\}$ stores a~collection of $k$ children data structures at level $q - 1$.

A~data structure $\D$ at level $q$ maintaining an~instance $I = (G, D, \cost, \supply, \demand)$ preserves the following invariants:

\begin{enumerate}[label=(I\arabic*)]
  \item \label{item:domset-struct-size} $G \in \Cc$ and $|V(G)| \leq n^{q / L}$;
  \item \label{item:domset-struct-domain} $I$ is $(\shat(q), \dhat(q))$-decent, for some functions $\shat$, $\dhat$ to be exactly specified later;
  \item \label{item:domset-struct-sol} $\D$ maintains a~nonnegative real $p$ satisfying $(1 - \delta \cdot \frac{q}{L})\OPT \le p \le \OPT$, where $\OPT$ is the minimum cost of a~solution to $I$.
\end{enumerate}

Each data structure at level $1$ maintains an~instance $I$ of size at most $n^{1 / L}$ and computes an~approximate lower bound on the cost of the solution within the $\frac{\delta}{L}$ fraction of $\OPT$, by rerunning the Baker's technique on each update.
By \cref{lem:domset-static-baker}, each update to such a~data structure takes time $n^{1/L} \cdot \shat(1) \cdot \dhat(1)^{\Oh(L/\delta)}$.

Now consider a~data structure $\D$ at level $q \ge 2$ maintaining a~$(\widehat{s}(q), \widehat{d}(q)$)-decent instance $I = (G, D, \cost, \supply, \demand)$ of \MinGeneralizedDom. We define a~variable $\Icur = (\Gcur, D^{\mathrm{cur}}, \cost^{\mathrm{cur}}, \supply^{\mathrm{cur}}, \demand^{\mathrm{cur}})$ tracking the current snapshot of $I$.
The lifetime of $\D$ is partitioned into \emph{epochs}: sequences of $\tau_q$ updates to $I$, with $\tau_q$ to be specified later.
The first epoch begins when $\D$ is initialized; and a~new epoch begins each time $\D$ processes $\tau_q$ updates to $I$.

Let $\Iold = (\Gold, D^{\mathrm{old}}, \cost^{\mathrm{old}}, \supply^{\mathrm{old}}, \demand^{\mathrm{old}})$ be the snapshot of $I$ at the start of each epoch; that is, at the start of an~epoch, we set $\Iold \coloneqq \Icur$.
We partition $V(\Gold)$ into layers $L_0, \dots, L_{4k-1}$ as follows: if $\Gold$ is connected, choose a~vertex $r \in V(G)$ and assigning a~vertex $v \in V(G)$ to $L_j$ if and only if $\mathsf{dist}_{\Gold}(r, v) \equiv j \mod{4k}$.
If $\Gold$ is disconnected, produce a~partitioning into layers for each connected components, and let $L_j$ be the union over $j$th layers for each connected component.

We now define pairwise disjoint dynamic sets $V_0, \dots, V_{k - 1}$.
For each $i \in \{0, \dots, k - 1\}$, let $V^{\mathrm{old}}_i \coloneqq L_{4i+1} \cup L_{4i+2}$ be the initial contents of $V_i$.
Thanks to this choice, all $V_i$-cleared subinstances of $I$ have small treewidth:

\begin{lemma}
  For each $i \in \{0, \dots, k - 1\}$, the treewidth of the Gaifman graph of $\Clear(\Iold; V_i)$ is bounded by $\Oh(k)$.
\end{lemma}
\begin{proof}
  Let $C$ be a~connected component of the Gaifman graph of $\Clear(\Iold; V_i)$.
  Note that $C$ cannot simultaneously contain vertices at distance $4k\ell + (4i + 1)$ and $4k\ell + (4i + 2)$, for any $\ell \in \N$.
  Therefore, $C$ is contained within the set of vertices of $\Clear(\Iinit; V_i)$ at distance at least $\ell$ and at most $\ell + 4k - 1$, for some $\ell \in \N$.
  By \cref{lem:csp-partitioning}, $C$ has treewidth bounded by $\Oh(k)$.
  Therefore, the entire Gaifman graph of $\Clear(\Iold; V_i)$ has treewidth bounded by $\Oh(k)$.
\end{proof}

Note that the sets $V_i$ are, contrary to the case of \CSP, dynamic -- when processing some updates, we may decide to shrink some of these sets.
Given that, let $i \in \{0, \dots, k - 1\}$, let $\Vcur_i$ track the current contents of $V_i$.
We maintain the following invariants about each set $\Vcur_i$:

\begin{itemize}
  \item For each $i \in \{0, \dots, k - 1\}$, we have that $\Vcur_i \subseteq V^{\mathrm{old}}_i$.
  \item The sets $N_{\Gcur}[\Vcur_0], \dots, N_{\Gcur}[\Vcur_{k-1}]$ are pairwise disjoint.
\end{itemize}

By the construction of $V_i$, the invariants are satisfied at the time of the initialization: we have $\Gcur = \Ginit$, so for each $i \in \{0, \dots, k - 1\}$ it holds that $N_{\Gcur}[\Vcur_i] \subseteq L_{4i} \cup L_{4i+1} \cup L_{4i+2} \cup L_{4i + 3}$.
Note that from the invariant it follows that the sets $\Vcur_i$ are pairwise disjoint.

Next, we maintain $k$ universes $\Icur_0, \dots, \Icur_{k-1}$, where $\Icur_i = \Clear(\Icur; \Vcur_i)$.
For each universe, we initialize a~data structure from \cref{lem:domset-dynamic-compression}, maintaining an~instance $I^\star_i$ that is a~compression of $\Icur_i$ with the same minimum cost of a~solution.
Finally, we initialize $k$ child data structures $\D_0, \dots, \D_{k-1}$ at level $q-1$, where each $\D_i$ stores the compressed instance $I^\star_i$ and maintains a~$(1 - \delta\frac{q-1}{L})$-approximation to the minimum cost of a~solution to $I^\star_i$.
As previously, this initialization is recursive: each child data structure maintains another collection of $k$ child data structures at level $q - 2$ each, etc., until construction $k^{q-1}$ data structures at level $1$ in total.

Assume each child data structure $\D_i$ maintains a~good lower bound $p_i$ on the minimum cost of a~solution to $\Icur_i$; then we can also preserve a~good lower bound on the minimum cost of a~solution to $I$ by keeping the maximum of the lower bounds:

\begin{lemma}
  \label{lem:dom-approximations}
  Let $\OPT_i$ be the minimum cost of a~solution to the instance $\Icur_i$, and let $p_i$ be so that $(1 - \delta\frac{q-1}{L})\OPT_i \le p_i \le \OPT_i$. Let $p = \max(p_0, \dots, p_{k-1})$.
  Then $(1 - \delta\frac{q}{L})\OPT \le p \le \OPT$, where $\OPT$ is the minimum cost of a~solution to $I$.
\end{lemma}
\begin{proof}
  Recall that the closed neighborhoods of the sets $\Vcur_0, \dots,\allowbreak \Vcur_{k-1}$ are pairwise disjoint.
  Thus, \cref{lem:domset-baker-partitioning} applies to the cleared subinstances $\Icur_0, \dots, \Icur_{k-1}$ and we get that:
  \begin{itemize}
    \item for every $j \in \{0, \dots, k-1\}$, we have that $\OPT_j \le \OPT$;
    \item there exists $j \in \{0, \dots, k-1\}$ such that $\OPT_j \ge (1 - \frac1k)\OPT$.
  \end{itemize}
  Thus, by our assumptions, for every $j \in \{0, \dots, k-1\}$ we have $p_j \le \OPT$, so also $p \le \OPT$.
  Moreover, there exists $j \in \{0, \dots, k-1\}$ with $\OPT_j \ge (1 - \frac1k)\OPT$.
  For this value of $j$, we have $p_j \ge (1 - \delta\frac{q-1}{L})\OPT_j$, so in total,
  $p_j \ge (1 - \delta\frac{q-1}{L})(1 - \delta\frac{1}{L})\OPT \ge (1 - \delta\frac{q}{L})\OPT$.
  As $p \ge p_j$, we conclude that $p \ge (1 - \delta\frac{q}{L})\OPT$.
\end{proof}

We now show how the data structure is updated.
First, we describe the removal of a~vertex $v \in \Vcur_i$ from $\Vcur_i$.
In this setting, the cleared subinstance $\Icur_i$ must be updated: each of the edges of $\Icur$ with one endpoint in $v$ and the other endpoint in $\Vcur_i \setminus \{v\}$ must be introduced to $\Icur_i$.
Moreover, since $v$ ceases to be relieved in $\Icur_i$, the function $\demand_v$ in $\Icur_i$ must be amended; this is done by removing and reinserting all already existing edges incident to $v$ with the updated demand sets.
Since $\Icur$ is $(\shat(q), \dhat(q))$-decent, the entire process involves $\Oh(\shat(q))$ updates to $\Icur_i$.
Each of these updates, in turn, causes the data structure from \cref{lem:domset-dynamic-compression} to apply a~number of updates to $I^\star_i$, which are then relayed to $\D_i$.

Now, we implement an~update to $\Icur$ (i.e, $\AddVertex$, $\AddEdge$, $\RemoveEdge$, or $\UpdateCost$).
For every vertex $v$ involved in the update (a~fresh vertex, an~endpoint of an~added or removed edge, or a~vertex with its cost updated), whenever $v$ belongs to some set $\Vcur_i$, we first remove $v$ from $\Vcur_i$ as described above.
Afterwards, we apply the update to each subinstance $\Icur_i$; again, the data structure from \cref{lem:domset-dynamic-compression} then issues a~number of updates to $I^\star_i$, which we relay to $\D_i$.
Finally, $\D$ recomputes a~lower bound on the minimum cost of a~solution to $\Icur$ by querying each $\D_i$ for the lower bound $p_i$ on the minimum cost of a~solution to $I^\star_i$ and returning the maximum value.

It remains to verify the satisfaction of the invariants.
First, we verify the invariants regarding the sets $\Vcur_0, \dots, \Vcur_{k-1}$.
The only non-trivial property is the preservation of the disjointness of the closed neighborhoods of the sets after the addition of an edge to $\Icur$:
\begin{lemma}
  Let $G$ be a~graph and $V_0, \dots, V_{k-1} \subseteq V(G)$ be so that $N_G[V_0], \dots, N_G[V_{k-1}]$ are pairwise disjoint.
  Let also $u, v \in V(G)$, and consider a~graph $G'$ equal to $G$ with an~edge $uv$ added, and $V'_i = V_i \setminus \{u, v\}$ for $i \in \{0, \dots, k-1\}$.
  Then $N_{G'}[V'_0], \dots, N_{G'}[V'_{k-1}]$ are pairwise disjoint.
\end{lemma}
\begin{proof}
  Let $i \in \{0, \dots, k-1\}$ and suppose there exists a~vertex $w \in N_{G'}[V'_i] \setminus N_G[V_i]$.
  Equivalently, $N_{G'}[w]$ intersects $V'_i$, but $N_G[w]$ is disjoint from $V_i$.
  If $w \notin \{u, v\}$, then $N_{G'}[w] = N_G[w]$ and $V'_i \subseteq V_i$ -- a~contradiction.
  Now assume without loss of generality that $w = u$.
  Then $N_{G'}[w] \subseteq N_G[w] \cup \{v\}$, but $V'_i \subseteq V_i \setminus \{v\}$ -- again a~contradiction.
  We conclude that $N_{G'}[V'_i] \subseteq N_G[V_i]$ for every $0 \le i < k$, so if the sets $N_G[V_0], \dots, N_G[V_{k-1}]$ were pairwise disjoint, then so are $N_{G'}[V'_0], \dots, N_{G'}[V'_{k-1}]$.
\end{proof}

Since each update to $\Icur$ causes at most $\Oh(\shat(q))$ updates to each $\Icur_i$, it follows that  during one epoch, the size of each $I^\star_i$ does not grow above $\tau_q \cdot (\shat(q) \cdot k)^{\Oh(1)} \log n$ (\cref{lem:domset-dynamic-compression}\ref{item:dom-compress-size}).
Later, we will choose $\tau_q$ so that this value is significantly less than $n^{(q-1)/L}$ so as to ensure the satisfaction of invariant \ref{item:indset-struct-size} by the children data structures; the fact that the Gaifman graph of $I^\star_i$ belongs to $\Cc$ follows from \cref{lem:domset-dynamic-compression}\ref{item:dom-compress-in-class}.
The invariant \ref{item:domset-struct-domain} is satisfied due to \cref{lem:domset-dynamic-compression}\ref{item:dom-compress-domain}, provided we set $\shat$ and $\dhat$ according to the bounds provided by \cref{lem:domset-dynamic-compression}.
Also, since the revenue of an optimum solution to $I^\star_i$ is equal to that of $\Icur_i$ (\cref{lem:domset-dynamic-compression}\ref{item:dom-compress-revenue}), each child data structure $\D_i$ maintains a~nonnegative real $p_i$ satisfying the preconditions of \cref{lem:dom-approximations}.
Therefore, the value $p \coloneqq \max(p_0, \dots, p_{k-1})$ is an~$(1 - \delta \frac{q}{L})$-approximation of the optimum revenue in $\Icur$, which proves the satisfaction of invariant \ref{item:domset-struct-sol} by~$\D$.
Finally, as previously, when an~epoch in $\D$ ends, $\D$ is reinitialized with $\Iold \coloneqq \Icur$ and a~new epoch starts, causing the destruction and the recursive reinitialization of the children data structures $\D_i$.

\subsection{Setting the parameters and time complexity analysis}

Here, we perform the complexity analysis of the data structure for \MinGeneralizedDom. That is, we prove the following statement:

\begin{lemma}
  \label{lem:domset-time-analysis}
  Suppose that the values $s$, $d$ are absolute constants.
  Then the parameters $L$, $\shat(\cdot)$, $\dhat(\cdot)$ and $\tau_{\cdot}$ can be chosen so that the data structure for \MinGeneralizedDom performs each update in time $f(\delta) \cdot n^{o(1)}$, where $f(\delta)$ is doubly-exponential in $\Oh(1 / \delta^2)$.
\end{lemma}

The proof is an~adaptation of \cref{ssec:indset-parameters} to the setting of \MinGeneralizedDom with slightly different bounds in several parts of the analysis.

\begin{proof}[Proof of \cref{lem:domset-time-analysis}]
  Recall that all constants depending on $\Cc$ are absolute constants.
  
  First, we bound the values of $\shat(\cdot)$ and $\dhat(\cdot)$.
  
  \begin{claim}
    \label{cl:domset-sd-parameters}
    One can set functions $\shat$ and $\dhat$ so that the invariant \ref{item:domset-struct-domain} is satisfied for all constructed data structures, and moreover $\shat(q) \in k^{\Oh(L)}$ and $\dhat(q) \in 2^{k^{\Oh(L)}}$ for all $q \in \{1, \dots, L\}$.
  \end{claim}
  \begin{claimproof}
    Recall that \ref{item:dom-compress-domain} states that the data structures at level $q$ maintain $(\shat(q), \dhat(q))$-decent instances of \MinGeneralizedDom.
    By the assumptions, it is enough to set $\shat(L) = s$ and $\dhat(L) = d$, i.e., $\shat(L)$ and $\dhat(L)$ are absolute constants.
    Assuming that a~given data structure $\D$ at level $q > 1$ maintains a~$(\shat(q), \dhat(q))$-decent instance $I$ of \MinGeneralizedDom, $\D$ stores a~collection of children data structures at level $q-1$, each maintaining a~compression of $I$ produced by {lem:domset-dynamic-compression}; from \cref{lem:domset-dynamic-compression}\ref{item:dom-compress-domain} we infer that each such compression is $(\Oh(\shat(q)w), \dhat(q) + 4^{\Oh(\shat(q)w)})$-decent, where $w \in \Oh(k)$.
    Choose an~absolute constant $C > 0$ so that for every $s', d' \geq 1$, a~compression of a~$(s', d')$-decent instance is $(C \cdot s'k, d' + 4^{C \cdot s'k})$-decent.
    Then it is enough to set $\shat(q - 1) = \shat(q) \cdot Ck$ and $\dhat(q - 1) = \dhat(q) + 4^{C \cdot \shat(q) k}$.
    A~straightforward induction implies that $\shat(q) = s \cdot (Ck)^{L - q}$.
    Therefore, assuming $k \ge 2$, we have $\shat(q) \in k^{\Oh(L)}$ and $\dhat(q) \in 2^{k^{\Oh(L)}}$ for all $q \in \{1, \dots, L\}$.
  \end{claimproof}
  
  Each auxiliary data structure at level $q$ spawns $k$ auxiliary data structures at level $q - 1$ for $q \ge 2$; and each auxiliary data structure at level $q = 1$ is a~leaf, i.e., it does not spawn any children data structures.
  Hence, there are at most $k^L$ auxiliary data structures at level $1$.
  Therefore:
  
  \begin{claim}
    \label{cl:domset-number-of-updates}
    Each update to $\D^{\mathrm{main}}$ causes at most $k^{\Oh(L^2)} (\log n)^{\Oh(L)}$ updates throughout all data structures.
  \end{claim}
  \begin{claimproof}
    By \cref{lem:domset-dynamic-compression}\ref{item:dom-compress-size}, each update to a~universe $\Icur_i$ causes at most $\shat(q) w^{\Oh(1)} \log n$ updates propagated to instances at level $q - 1$.
    Since each instance at level $q > 1$ has $k$ associated universes $\Icur_i$ and each update to $\Icur$ causes at most $\Oh(\shat(q))$ updates to each universe $\Icur_i$, we conclude that a~single update to an~instance maintained by a~data structure at level $q$ causes at most
    \begin{equation}
    \label{eq:domset-number-of-updates}
    \shat(q)^2 w^{\Oh(1)} \log n
    \end{equation}
    updates to be propagated to the children data structures at level $q - 1$.
    Recall that $w \in \Oh(k)$ and that by \cref{cl:domset-sd-parameters} we have $\shat(q) \in k^{\Oh(L)}$.
    Therefore, (\ref{eq:domset-number-of-updates}) is bounded by
    \[ k^{\Oh(L)} \log n. \]
    By straightforward induction, a~single update to $\D^{\mathrm{main}}$ (the data structure at level $L$) causes at most
    \[ \Oh\left(\left(k^{\Oh(L)} \log n\right)^L\right) \leq k^{\Oh(L^2)} (\log n)^{\Oh(L)}. \]
    updates to all data structures.
  \end{claimproof}  
  
  \begin{claim}
    \label{cl:domset-time-without-reinitializations}
    The total time required to update all data structures for a~single update to $\D^{\mathrm{main}}$, excluding the time of all required reinitializations, is bounded by $n^{1/L} \cdot 2^{(\frac{L}{\delta})^{\Oh(L)} + \Oh(L \log \log n)}$.
  \end{claim}
  \begin{claimproof}
    Focus first on the time required to recompute the solutions to the instances at level $1$.
    Each data structure at level $1$ recomputes the solution from scratch using \cref{lem:domset-static-baker} with an~instance with at most $n^{1/L}$ vertices, the bound on the maximum degree $s = \shat(1)$, the bound on the domain size $d = \dhat(1)$ and the approximation guarantee $\frac{\delta}{L}$.
    Therefore, each update at level $1$ is processed in time $n^{1/L} \cdot \shat(1) \cdot \dhat(1)^{\Oh(L / \delta)} \leq n^{1 / L} \cdot 2^{k^{\Oh(L)} \cdot \frac{L}{\delta}}$.
    At a~data structure at level $q > 1$, it can be verified that each update is processed in time $2^{\Oh(\shat(q) k)} \cdot (\log n)^{\Oh(1)}$ (excluding the time spent at data structures at lower levels).
    Since $2^{\Oh(\shat(q) k)} \cdot (\log n)^{\Oh(1)} \leq 2^{k^{\Oh(L)}} \cdot (\log n)^{\Oh(1)}$, we conclude that for any fixed auxiliary data~structure, every update is processed in time at most $n^{1/L} \cdot (\log n)^{\Oh(1)} \cdot 2^{k^{\Oh(L)} \cdot \frac{L}{\delta}}$.
        
    Using the bound from \cref{cl:domset-number-of-updates} on the total number of updates to all data structures, we infer that, for a~single update to $\D^{\mathrm{main}}$, the total recomputation time in the data structures is bounded by $k^{\Oh(L^2)} (\log n)^{\Oh(L)} n^{1/L} (\log n)^{\Oh(1)} 2^{k^{\Oh(L)} \cdot \frac{L}{\delta}}$.
    Recall that we set $k = \left\lceil \frac{L}{\delta} \right\rceil$, so we have that $k^{\Oh(L^2)} = 2^{\Oh(L^2 \log \frac{L}{\delta})}$, $(\log n)^{\Oh(L)} = 2^{\Oh(L \log \log n)}$, and $2^{k^{\Oh(L)} \cdot \frac{L}{\delta}} = 2^{\left(\frac{L}{\delta}\right)^{\Oh(L)}}$.
    Therefore, we can bound the total update time, excluding the reinitializations, by 
    \[n^{1/L} \cdot 2^{\Oh(L^2 \log \frac{L}{\delta} + L \log \log n + (\frac{L}{\delta})^{\Oh(L)})} \leq
    n^{1/L} \cdot 2^{(\frac{L}{\delta})^{\Oh(L)} + \Oh(L \log \log n)}.\]
  \end{claimproof}
  
  Now we set the epoch lengths and bound the amortized time of all reinitializations per a~single update to $\D^{\mathrm{main}}$.
  Let $\tau_q$ denote the epoch length for data structures on level $q$ for some $2 \le q \le L$.
  Recall that the epoch length for a~data structure is measured in the number of updates to this particular instance (as opposed to $\D^{\mathrm{main}}$). Recall also from \cref{item:dom-compress-size} of \cref{lem:domset-dynamic-compression} that each update to a~data structure $\D$ at level $q$ generates $(\shat(1) k \log n)^{\Oh(1)} \leq k^{\Oh(L)} (\log n)^{\Oh(1)}$ updates to the children structures.
  Fix now a~constant $c > 0$ such that this number of updates is actually bounded by $(k^L \log n)^c$, and set $\tau_q \coloneqq n^{(q - 1) / L} / (k^L \log n)^c$.
  Then, \cref{lem:domset-dynamic-compression}\ref{item:dom-compress-size} implies invariant \ref{item:domset-struct-size}: children data structures at level $q-1$ are guaranteed to maintain instances of size $n^{(q - 1)/L}$.
  
  We now bound the amortized time complexity of initializations and reinitializations across all data structures.
  
  \begin{claim}
    \label{cl:domset-time-only-reinitializations}
    The amortized time of all initializations and reinitializations per a~single update to $\D^{\mathrm{main}}$ is bounded by $n^{1/L} \cdot 2^{(\frac{L}{\delta})^{\Oh(L)} + \Oh(L \log \log n)}$.
  \end{claim}
  \begin{claimproof}
    Consider an~initialization or a~reinitialization of a~data structure $\D$, maintaining an~instance $I$ of \MinGeneralizedDom, at some level $q \ge 2$.
    The (re)initialization takes time $|V(I)| \log |V(I)| \cdot 2^{k^{\Oh(L)}}$ by \cref{lem:domset-dynamic-compression} and \cref{cl:domset-sd-parameters}.
    Thanks to the invariants, we have $|V(G)| \le n^{q/L}$ and $|E(G)| = O(|V(G)|)$, hence the (re)initialization time is bounded by $n^{q/L} \log n \cdot 2^{k^{\Oh(L)}}$.
    Note that the (re)initialization of $\D$ causes all children data structures to be destructed and reset to constant-sized instances; however, this can be done in amortized constant time per instance.
    
    Consider the state of $\D$ after some $u$ updates; let $t = \left\lceil u / \tau_q \right\rceil$ be the index of the current epoch of $\D$, so that the lifetime of $\D$ is assumed to contain $t - 1$ full epochs so far.
    Then, the total time of all reinitializations of $\D$ so far is bounded by $(t - 1) \cdot n^{q/L} \log n \cdot 2^{k^{\Oh(L)}} \leq (t - 1) \cdot \tau_q \cdot n^{1/L} \cdot (k^L \log n)^{c+1} \cdot 2^{k^{\Oh(L)}} \leq u \cdot n^{1/L} (\log n)^{\Oh(1)} \cdot 2^{k^{\Oh(L)}}$.
    Hence, the amortized time for the reinitializations of $\D$ can be bounded by $n^{1/L} \cdot (\log n)^{\Oh(1)} \cdot 2^{k^{\Oh(L)}}$ per an~update to $\D$.
    
    As each update to $\D^{\mathrm{main}}$ causes $k^{\Oh(L^2)} (\log n)^{\Oh(L)}$ updates to all auxiliary data structures in total and there are at most $k^L$ auxiliary data structures, we conclude that the total amortized time required for all reinitializations per a~single update to $\D^{\mathrm{main}}$ is bounded by \[k^{\Oh(L^2)} \cdot n^{1/L} \cdot (\log n)^{\Oh(L)} \cdot 2^{k^{\Oh(L)}} \leq n^{1/L} (\log n)^{\Oh(L)} \cdot 2^{(\frac{L}{\delta})^{\Oh(L)}} = n^{1/L} \cdot 2^{(\frac{L}{\delta})^{\Oh(L)} + \Oh(L \log \log n)}.\]
  \end{claimproof}
  
  Note that \cref{cl:domset-time-without-reinitializations,cl:domset-time-only-reinitializations} provide the same time complexity bounds as \cref{lem:time1,lem:time2} in the setting of \CSP.
  Therefore, the following final time complexity analysis is shown by the same argument as \cref{lem:time3}, hence we omit its proof.
  \begin{claim}
    $L$ may be set so that the amortized time of an~update to $\D^{\mathrm{main}}$ is $n^{\Oh\left(\frac{\log \log \log n}{\log \log n \cdot \delta}\right)} = n^{o\left(\frac{1}{\delta}\right)}$.
  \end{claim}
  Finally, as before, we use the following trick to separate $\delta$ from $n$ in the final time complexity.
  If $\log \log n > \frac{1}{\delta^2}$, then $\delta > \sqrt{\log \log n}$, hence $n^{\Oh\left(\frac{\log \log \log n}{\log \log n \cdot \delta}\right)} = n^{\Oh\left(\frac{\log \log \log n}{\sqrt{\log \log n}}\right)} = n^{o(1)}$.
  In the opposite case, we have $n \le 2^{2^{1 / \delta^2}}$; in this case, on every update, we rerun the static algorithm from \cref{lem:domset-static-baker} in time $n s d^{\Oh(1 / \delta)} \leq 2^{2^{\Oh(1 / \delta^2)}}$.
  In any case, the time complexity of an~update is bounded by $f(\delta) \cdot n^{\Oh\left(\frac{\log \log \log n}{\sqrt{\log \log n}}\right)}$, where $f(\delta) \in 2^{2^{\Oh(1 / \delta^2)}}$.
  This finishes the proof.
\end{proof}

%
%
%
%

\bibliographystyle{plain}
\bibliography{references}

\appendix

\end{document}